\documentclass[journal]{IEEEtran}

\usepackage{setspace}
\usepackage{amsmath}
\usepackage{amsfonts}
\usepackage{graphicx}
\usepackage{amsthm}
\usepackage{amssymb}
\usepackage{mathrsfs}
\usepackage{stfloats}
\usepackage{geometry}
\usepackage{color}
\usepackage{bm}

\newtheorem{lemma}{Lemma}
\newtheorem{theorem}{Theorem}

\begin{document}
%

\title{Performance Analysis of Multi-Cell Millimeter Wave Massive MIMO Networks with Low-Precision ADCs}


\author{
\IEEEauthorblockN{Jindan Xu$^{\ast}$,~\emph{Student Member,~IEEE}, Wei Xu$^{\ast}$,~\emph{Senior Member,~IEEE}, Hua Zhang$^{\ast}$,~\emph{Member,~IEEE}, \\
 Geoffrey Ye Li$^{\dag}$,~\emph{Fellow,~IEEE}, and Xiaohu You$^{\ast}$,~\emph{Fellow,~IEEE}}\\
\IEEEauthorblockA{
$^{\ast}$National Mobile Communications Research Laboratory, Southeast University, Nanjing 210096, China\\
$^{\dag}$School of Electrical and Computer Engineering, Georgia Institute of Technology, Atlanta, GA 30332, USA\\
Email: \{jdxu, wxu, huazhang\}@seu.edu.cn, liye@ece.gatech.edu, xhyu@seu.edu.cn}
\vspace{-0.8cm}
\thanks{
Manuscript received April 27, 2018; revised August 17, 2018; accepted September 29, 2018.
This work of was supported in part by NSFC under grants 61871109, 61601115, 61571118, and U1534208, the Open Research Fund of the State Key Lab of ISN under ISN18-03, and the Six Talent Peaks project in Jiangsu Province under GDZB-005.
The editor coordinating the review of this paper and approving it for publication was X. Wang.
\textit{(Corresponding author: Wei Xu).}

J. Xu and W. Xu are with the National Mobile Communications Research Laboratory (NCRL), Southeast University, Nanjing 210096, China, and also with the State Key Laboratory of Integrated Services Networks, Xidian University, Xi'an 710071, China (jdxu@seu.edu.cn; wxu@seu.edu.cn).

H. Zhang and X. You are with the National Mobile Communications Research Laboratory (NCRL), Southeast University, Nanjing 210096, China (huazhang@seu.edu.cn; xhyu@seu.edu.cn).

Geoffrey Ye Li is with the School of Electrical and Computer Engineering, Georgia Institute of Technology, Atlanta, GA 30332, USA (liye@ece.gatech.edu).

}

}


\maketitle

\begin{abstract}
In this paper, we investigate a multi-cell millimeter wave (mmWave) massive multiple-input multiple-output (MIMO) network with low-precision analog-to-digital converters (ADCs) at the base station (BS).
Each cell serves multiple users and each user is equipped with multiple antennas but driven by a single RF chain.
We first introduce a channel estimation strategy for the mmWave massive MIMO network and analyze the achievable rate with imperfect channel state information.
Then, we derive an insightful lower bound for the achievable rate, which becomes tight with a growing number of users.
The bound clearly demonstrates the impacts of the number of antennas and the ADC precision, especially for a single-cell mmWave network at low signal-to-noise ratio (SNR). It characterizes the tradeoff among various system parameters. Our analytical results are finally confirmed by extensive computer simulations.

\end{abstract}


\IEEEpeerreviewmaketitle

\section{Introduction}
Massive multiple-input multiple-output (MIMO) is a critical technique to significantly improve the performance of the fifth generation (5G) cellular network \cite{MIMO1}.
In massive MIMO, the base station (BS) is equipped with hundreds, or even thousands, of antennas to provide high spectral and power efficiency.
However, both cost and power consumption increase dramatically with the number of antennas, partly because each antenna requires a pair of dedicated analog-to-digital converters (ADCs). Fortunately, there are two potential means of alleviating this challenging issue.
On one hand, low-precision ADCs can be employed since the power consumption decreases exponentially with the quantization precision \cite{ADC11}-\cite{ADC6}.
An overview on channel estimation, signal detector, and transmit precoding for massive MIMO using low-precision ADCs in future networks has been provided in \cite{ADC_magazine}.
Specifically in \cite{ADC1}, it has shown that 1-bit ADCs can achieve satisfactory performance in terms of theoretical capacity and symbol error rate (SER) in massive MIMO uplink systems.
Furthermore, the spectral efficiencies of a mixed-ADC system under energy constraint has been studied in \cite{ADC Mix}.
The mixed-ADC architecture in frequency-selective channels has been investigated in \cite{ADC Mix3}.
It has been demonstrated in \cite{ADC7} that low-precision, e.g., 2-3 bits, ADCs only cause limited sum rate loss under some mild assumptions for an amplify-and-forward relay uplink network.
Studies in \cite{ADC3} and \cite{ADC4} have analyzed the performance of low-precision transceivers in multiuser massive MIMO downlinks.
On the other hand, radio-frequency (RF) chains can be also constrained to reduce the total number of required converters, which leads to a hybrid transceiver architecture \cite{Hybrid} \cite{Hybrid4}.
A low-complexity hybrid precoding method has been proposed in \cite{Hybrid_Liang}.
The study in \cite{Hybrid_Yu} has shown that hybrid beamforming can asymptotically approach the performance of fully digital beamforming for a sufficiently large number of antennas.
However, in many scenarios, low-precision ADCs inevitably deteriorate the performance while the architecture with limited RF chains sacrifices the multiplexing gain.
In practice, it is interesting to find a cost-efficiency tradeoff when employing low-precision ADCs and a limited number of RF chains \cite{ADC Fan} \cite{Hybrid5}.

Meanwhile, in order to achieve ultra high data rates, the spectrum ranging from 30 GHz to 300 GHz, namely millimeter wave (mmWave), looks attractive in 5G \cite{mmWave1}.
The ten-fold increase in carrier frequency, compared to the current majority of wireless systems, implies that mmWave signals experience an order-of-magnitude increase in free-space loss \cite{mmWave_propagation1} \cite{mmWave_propagation2}.
Fortunately, the decrease in wavelength enables to pack a large number of antenna elements into small form factors.
Large antenna arrays in mmWave systems are leveraged to combat severe pathloss through a large beamforming gain \cite{mmWave2}.
In \cite{mmWave3} and \cite{mmWave4}, hybrid beamforming has been investigated in mmWave MIMO networks.
A joint beam selection scheme for analog precoding has been proposed in \cite{mmWave_hybrid1} and a relay hybrid precoding design has been studied in  \cite{mmWave_hybrid2}.
Different from conventional wireless channels in cellular networks, spatial sparsity emerges as a dominant nature in mmWave propagations \cite{Sparsity1}.
By exploiting the sparsity, a beamforming training algorithm has been proposed in \cite{Sparsity2}.
Then, random beamforming has been studied in \cite{Sparsity3} as well as a user scheduling algorithm proposed for beam aggregation.
For low-complexity hybrid precoding, an algorithm using generalized orthogonal matching pursuit has been proposed in \cite{mmWave_hybrid3} when the knowledge of channel sparsity is known.


It is known that the availability of channel state information (CSI) plays a critical role in beamforming design \cite{MIMO3}.
A two-stage precoding scheme has been proposed in \cite{MIMO_add1} to reduce the overhead of both channel training and CSI feedback in massive MIMO systems.
Furthermore, an interference alignment and soft-space-reuse based cooperative transmission scheme has been proposed in \cite{MIMO_add2} and a low-cost channel estimator has been designed.
For systems with low-precision ADCs, conventional pilot-aided channel estimation can in some scenarios be used to acquire the CSI \cite{CE0}.
However, it can be hardly applied to the multiuser hybrid system because the number of RF chains is much smaller than the antenna number.
Therefore, channel estimation using overlapped beam patterns and rate adaptation has been proposed in \cite{CE3} and a limited feedback hybrid channel estimation has been studied in \cite{CE4}.
To overcome the drawback of the feedback-based mechanism in these methods, a low-complexity channel estimation method has been proposed in \cite{CE1}.

In this paper, we investigate a non-cooperative multi-cell mmWave system with a large-scale antenna array, where low-precision ADCs are used at the BS.
We assume that each cell serves multiple users and each user is equipped with multiple antennas but driven by a single RF chain.
Analog beamforming is therefore conducted at user sides based on the estimated CSI.
Most of the existing works, like \cite{ADC_magazine}-\cite{Hybrid_Yu}, focused on either low-precision quantization or hybrid architecture.
In our work, we study both the low-precision ADCs at the BS and analog beamforming at the user side.
This setup is of much interest due to its implementational popularity in practice \cite{ADC Fan} \cite{Hybrid5}.
To the best of our knowledge, there is few works investigating the performance of the network \cite{Both}.
Main contributions of this work are summarized as follows:

1)
We derive the ergodic achievable rate of the network with imperfect CSI by using an ADC quantization model based on the Bussgang theorem.
Although the popular tools, such as the law of large numbers and the central limit theorem, do not apply here due to the sparsity of mmWave channel,
we successfully derive a lower bound for the user ergodic rate with the help of stochastic calculations.

2)
Based on the derived lower bound, the impacts of various system parameters, including the ADC precision, signal and pilot SNRs, and the numbers of users and antennas, on the system performance have been characterized.
A typical scenario of a single-cell network is investigated by retrieving as a special case from our derived results.
We find that the received signal-to-interference-quantization-and-noise ratio (SIQNR) can be expressed as a scaling value of the original low SNR.

The rest of this paper is organized as follows.
Both ADC quantization and mmWave channel models are described in Section~\uppercase\expandafter{\romannumeral2}.
In Section~\uppercase\expandafter{\romannumeral3}, we introduce a two-step channel estimation method for the multi-cell hybrid system.
In Section~\uppercase\expandafter{\romannumeral4}, we analyze the achievable uplink rate with imperfect CSI and low-precision quantization error and derive a lower rate bound.
Then based on the bound, we analyze the performance under two special scenarios in Section~\uppercase\expandafter{\romannumeral5}.
Simulation results are presented in Section \uppercase\expandafter{\romannumeral6} and conclusions are drawn in Section~\uppercase\expandafter{\romannumeral7}.

\emph{Notations}: $\textbf{A}^T$, $\textbf{A}^*$ and $\textbf{A}^H$ represent the transpose, conjugate and conjugate transpose of $\textbf{A}$, respectively.
$\textbf{a}_{i}$ represents the $i$th column of $\textbf{A}$.
$\textrm{diag}(\textbf{A})$ keeps only the diagonal elements of $\textbf{A}$, while $\textrm{diag}\{a_1, a_2,...,a_N\}$ generates a diagonal matrix with entries $a_1, a_2,...,a_N$.
$\mathbb{E}\{\cdot\}$ is the expectation operator.
U$[a,b]$ denotes the uniform distribution between $a$ and $b$.
$\longrightarrow$ denotes the almost sure convergence.

\section{System Model}

\begin{figure*}[tb]
\centering\includegraphics[width=0.8\textwidth,bb=0 90 960 480]{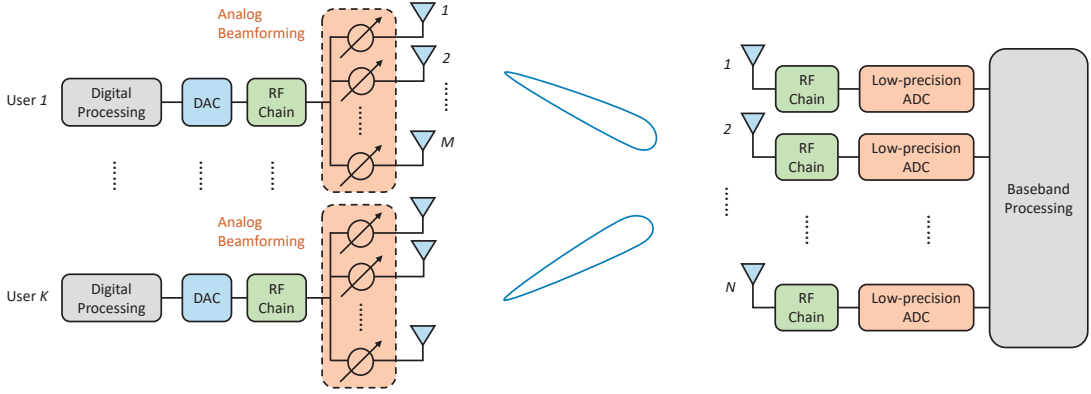}
\caption{Block diagram of the multiuser massive MIMO system in each cell.}
\label{block}
\end{figure*}

We consider a non-cooperative multi-cell system consisting of $L$ cells. 
In each cell, $K$ user terminals are served simultaneously and $N$ antennas are equipped at the BS.
Universal frequency reuse is exploited, and therefore both intra-cell and inter-cell interferences exist.

\subsection{Quantization Model for Low-Precision ADCs}

As in Fig. \ref{block}, each user equips $M$ antennas driven by a single RF chain. The RF chain can access to all the $M$ antennas through $M$ phase shifters, which allows analog beamforming for both transmitting and receiving.
At the BS side, a pair of low-precision ADCs is exploited for each antenna for processing the in-phase and quadrature input signals.

It is in general difficult to accurately analyze the signal quantization error of low-precision ADCs. Fortunately, an approximately linear representation has been widely adopted by using the Bussgang theorem \cite{Bus1}.
This quantization model has been verified accurate enough for characterizing commonly used ADCs, especially for popular quantization levels in practice \cite{ADC3} \cite{rho}.
It decomposes the ADC quantization into two uncorrelated parts as
\begin{equation}
\label{ADC}
\mathcal{Q}_{AD}(\textbf{y})=\textbf{F}\textbf{y}+\textbf{n}_{q},
\end{equation}
where $\mathcal{Q}_{AD}(\cdot)$ is the quantization operation of ADC, $\textbf{y}\in \mathbb{C}^{N\times 1}$ denotes the vector before quantization, $\textbf{F}$ represents the quantization processing matrix, and $\textbf{n}_{q}\sim \mathcal{CN}(\mathbf{0},\sigma_{q}^2\textbf{I}_N)$ denotes the quantization noise.
From  \cite{Bus2} \cite{ADC2}, it follows
\begin{equation}
\label{F}
\textbf{F}=(1-\rho_{AD})\textbf{I}_N,
\end{equation}
and
\begin{equation}
\label{C_nq}
\mathbb{E}\{\textbf{n}_{q}{\textbf{n}_{q}}^H\}=\rho_{AD}(1-\rho_{AD})\textrm{diag}\left(\mathbb{E}\{\textbf{y}\textbf{y}^H\}\right),
\end{equation}
where $\rho_{AD}$ represents the distortion factor. The distortion factor depends on the ADC precision, $b$, representing the number of the quantized bits of the ADC.

\subsection{Channel Model with Hybrid Architecture}
The uplink channel matrix from the $k$th user in the $l$th cell to the $j$th BS, $\textbf{H}_{jlk} \in \mathbb{C}^{N\times M}$, can be expressed as \cite{Sparsity1} \cite{channel}
\begin{equation}
\label{H}
\textbf{H}_{jlk}=\beta_{jlk}^{\frac{1}{2}}\textbf{h}_{B,jlk}\textbf{h}_{U,jlk}^H,
\end{equation}
where $\beta_{jlk}$ denotes the large-scale fading from the $k$th user in cell $l$ to BS $j$,
$\textbf{h}_{U,jlk}\in \mathbb{C}^{M\times 1}$ and $\textbf{h}_{B,jlk}\in \mathbb{C}^{N\times 1}$ denote the antenna array response vectors of the $k$th user in cell $l$ and the BS in cell $j$, respectively.
The small-scale fadings are represented by $\textbf{h}_{U,jlk}$ and $\textbf{h}_{B,jlk}$.
Due to the sparsity of mmWave channels, each of $\textbf{h}_{U,jlk}$ and $\textbf{h}_{B,jlk}$ is in general a single line-of-sight (LoS) path depending on the corresponding angle of incidence.
In particular, we have
\begin{equation}
\label{H_U}
\textbf{h}_{U,jlk}=[1,e^{-j2\pi\frac{d}{\lambda}\cos\varphi_{jlk}},...,e^{-j2\pi(M-1)\frac{d}{\lambda}\cos\varphi_{jlk}}]^T,
\end{equation}
and
\begin{equation}
\label{H_B}
\textbf{h}_{B,jlk}=[1,e^{-j2\pi\frac{d}{\lambda}\cos\theta_{jlk}},...,e^{-j2\pi(N-1)\frac{d}{\lambda}\cos\theta_{jlk}}]^T,
\end{equation}
where $\varphi_{jlk}\sim$ U$[0,\pi]$ and $\theta_{jlk}\sim$ U$[0,\pi]$ are the corresponding angles of incidence at the antenna arrays of user $k$ in cell $l$ and BS $j$, respectively, $d$ is the distance between adjacent antennas, and $\lambda$ is the wavelength of radio signals at the carrier frequency.
Typically, let $d=\frac{\lambda}{2}$ to minimize the space occupied by the massive antenna array while still achieving the optimal diversity \cite{channel}.

Since each user is equipped with a single RF chain even if it has multiple antennas as illustrated in Fig. \ref{block}, analog beamforming is conducted at the user sides.
Let $\textbf{w}_{lk}\in \mathbb{C}^{M\times 1}$ be the beamforming vector at user $k$ in cell $l$, which is determined by the estimated angle of arrival (AoA) to be elaborated in Section \uppercase\expandafter{\romannumeral3}.A.
The equivalent uplink channel from users in cell $l$ to BS $j$ can be expressed as
\begin{align}
\bar{\textbf{H}}_{jl}&=[\textbf{H}_{jl1}\textbf{w}_{l1},\textbf{H}_{jl2}\textbf{w}_{l2},...,\textbf{H}_{jlK}\textbf{w}_{lK}]
\nonumber
\\
&=\left[\beta_{jl1}^{\frac{1}{2}}c_{jl1}\textbf{h}_{B,jl1},\beta_{jl2}^{\frac{1}{2}}c_{jl2}\textbf{h}_{B,jl2},...,\beta_{jlk}^{\frac{1}{2}}c_{jlK}\textbf{h}_{B,jlK}\right],
\label{Heq}
\end{align}
where $c_{jlk}$ represents the beamforming gain of the $k$th user, defined as
\begin{equation}
\begin{aligned}
\label{c}
c_{jlk}\triangleq\textbf{h}_{U,jlk}^H\textbf{w}_{lk}, ~k=1,2,...,K.
\end{aligned}
\end{equation}

In the uplink, the received analog signals at BS are converted by ADCs before detection.
After the ADC operation, maximal ratio combining (MRC) is utilized for signal detection. At BS $j$, the received signal, $\textbf{y}_j\in \mathbb{C}^{K\times 1}$, can be expressed as
\begin{equation}
\label{y_ul}
\textbf{y}_j=\hat{\textbf{H}}_{jj}^H\mathcal{Q}_{AD}\left(\sqrt{P_t}\sum_{l=1}^L\bar{\textbf{H}}_{jl}\textbf{x}_l+\textbf{n}_j\right),
\end{equation}
where $\textbf{x}_l\in \mathbb{C}^{K\times 1}$ is normalized uplink data from all $K$ users in cell $l$, i.e., $\mathbb{E}\{\textbf{x}_l{\textbf{x}_l}^{H}\}=\textbf{I}_{K}$, $P_t$ is the transmit power of each user, and $\textbf{n}_j\sim \mathcal{CN}(\mathbf{0},\sigma_n^2\textbf{I}_{N})$ denotes the additive white Gaussian noise (AWGN) in cell $j$ with $\sigma_n^2$ representing the noise power.
Here, $\hat{\textbf{H}}_{jj}$ denotes the estimate of equivalent channel $\bar{\textbf{H}}_{jj}$ within cell $j$, which will be discussed later in Section~\uppercase\expandafter{\romannumeral3}.B.

By substituting the ADC model in \eqref{ADC} to \eqref{y_ul}, the detected signal at BS $j$ can be expressed as
\begin{equation}
\begin{aligned}
\label{y_ul_2}
\textbf{y}_{j}\!=(\!1\!-\!\rho_{AD}\!)\sqrt{P_t}\hat{\textbf{H}}_{jj}^H\sum_{l=1}^L\bar{\textbf{H}}_{jl}\textbf{x}_l
\!+\!(1\!-\!\rho_{AD})\hat{\textbf{H}}_{jj}^H\textbf{n}_j\!+\!\hat{\textbf{H}}_{jj}^H\textbf{n}_{q,j}\!,
\end{aligned}
\end{equation}
where $\textbf{n}_{q,j}\sim \mathcal{CN}(\mathbf{0},{\sigma_{q,j}}^2 \textbf{I}_N)$ denotes the quantization noise at BS $j$ arising from low-precision ADCs.
From \eqref{C_nq}, \eqref{H_B}, \eqref{Heq} and \eqref{y_ul}, we have
\begin{equation}
\label{sigma_q}
{\sigma_{q,j}}^2=\rho_{AD}(1-\rho_{AD})\left(\sigma_n^2+P_t\sum_{l=1}^L\sum_{k=1}^K\beta_{jlk} |c_{jlk}|^2\right).
\end{equation}

\section{Channel Estimation}
In the above communication process, a critical procedure includes determining the analog beamforming vector, $\textbf{w}_{lk}$, at user side and the digital combining matrices, $\hat{\textbf{H}}_{jj}$, at the BS. The design relies on the availability of CSI at the corresponding nodes.

\subsection{AoA Estimation}
In order to determine the beamforming vectors, we first estimate the angle of incidence at each user from the BS in the same cell. 
A similar procedure as in \cite{CE1} is introduced.
In order to avoid inter-cell interference, each cell conducts this step in an orthogonal way.
Taking cell $l$ for instance, BS $l$ broadcasts a frequency tone $x=\cos2\pi ft$ from an arbitrary antenna to all users.
Assuming channel reciprocity, the received signal at user $k$ in cell $l$ can be expressed as
\begin{equation}
\label{AoA}
r_{lk}=\beta_{llk}^{\frac{1}{2}}\tilde{\textbf{w}}_{lk}^T\textbf{h}_{U,llk}^* x +\tilde{\textbf{w}}_{lk}^T\textbf{n}^{\textrm{A}}_{lk},
\end{equation}
where $\textbf{n}^{\textrm{A}}_{lk}$ is the AWGN at user $k$ in cell $l$. The receiving beamforming vector, $\tilde{\textbf{w}}_{lk}$, is expressed as
\begin{equation}
\label{wk}
\tilde{\textbf{w}}_{lk}=\frac{1}{\sqrt{M}}\left[1,e^{-j2\pi\frac{d}{\lambda}\cos\tilde{\varphi}_{lk}},...,e^{-j2\pi(M-1)\frac{d}{\lambda}\cos\tilde{\varphi}_{lk}}\right]^T,
\end{equation}
where $\tilde{\varphi}_{lk}$ is the phase shift of the receiving antenna array.
To estimate AoA, we resort to choosing the optimal $\tilde{\varphi}_{lk}$ to maximize the power of received signal $r_{lk}$.
Since ideal analog phase shifter with continuous phase is less practical, we consider an analog beamformer of limited resolution. The value of $\tilde{\varphi}_{lk}$ is chosen from a codebook:
\begin{equation}
\label{psi}
\bm{\psi}=\left[\zeta,~3\zeta,~5\zeta,...,(2^{B+1}-1)\zeta\right],
\end{equation}
where $\zeta=\frac{\pi}{2^{B+1}}$ and $B$ is the number of quantization bits for phases.
Then, the estimated AoA of user $k$ in cell $l$ is chosen as
\begin{equation}
\label{AoA2}
\hat{\varphi}_{lk}=\arg \max\limits_{ \tilde{\varphi}_{lk} \in \bm{\psi}} |r_{lk}|.
\end{equation}
After obtaining $\hat{\varphi}_{lk}$, the beamforming vector for user $k$ in cell $l$ is accordingly set to:
\begin{equation}
\label{wk2}
\textbf{w}_{lk}=\frac{1}{\sqrt{M}}\left[1,e^{-j2\pi\frac{d}{\lambda}\cos\hat{\varphi}_{lk}},...,e^{-j2\pi(M-1)\frac{d}{\lambda}\cos\hat{\varphi}_{lk}}\right]^T.
\end{equation}
Given that $\textbf{w}_{lk}$ is determined, beamforming gain from BS $j$ to user $k$ in cell $l$, i.e., $c_{jlk}$ in \eqref{c}, can be obtained by measuring the received signal power. 
Note that the above AoA estimation is conducted in the downlink while the obtained analog beamforming vector, $\textbf{w}_{lk}$, is used for uplink transmission.
This is realizable thanks to the assumption of channel reciprocity in time division duplex (TDD) mode.

\subsection{Demodulation Channel Estimation}
With $\textbf{w}_{lk}$ in \eqref{wk2} and the obtained beamforming gain $c_{jlk}$, we can estimate the uplink channel by transmitting orthogonal pilots from users to the BS at this step.
After analog beamforming, we only need to estimate equivalent channel $\bar{\textbf{H}}_{jl}$ in \eqref{Heq}, instead of the original $\textbf{H}_{jlk}~(k=1,2,...,K)$ in \eqref{H} with a much larger size. Thus, the number of the required pilots decreases from $MK$ to $K$ and the dimension of matrix computation is greatly reduced.
If all cells reuse the same pilot sequences, pilot contamination should also be considered. 
Since low-precision ADCs are deployed at the BS, the accuracy of channel estimation is also affected by the ADC quantization.

Let user $k$ in each cell send pilot vector $\bm{\phi}_k\in \mathbb{C}^{\tau\times 1}$ where $\tau\geq K$ is the pilot length, which is orthonormal for different users. Define the pilot matrix, $\bm{\Psi}\in \mathbb{C}^{\tau\times K}$, as
\begin{equation}
\label{pilot}
\bm{\Psi}=[\bm{\phi}_1,\bm{\phi}_2,...,\bm{\phi}_K],
\end{equation}
then $\bm{\Psi}^H\bm{\Psi}=\textbf{I}_K$.
The received pilot signal at the $j$th BS before ADCs equals
\begin{equation}
\label{yp}
\textbf{Y}_{p,j}=\sqrt{P_p} \sum_{l=1}^L\bar{\textbf{H}}_{jl} \bm{\Psi}^T+\textbf{n}_{p,j},
\end{equation}
where $P_p$ is the pilot power and $\textbf{n}_{p,j}=[\textbf{n}_{p,j1},\textbf{n}_{p,j2},...,\textbf{n}_{p,j\tau}]$ denotes the AWGN with $\textbf{n}_{p,ji}\sim\mathcal{CN}(\mathbf{0},\sigma_n^2\textbf{I}_N)$ for $i=1,2,...,\tau$.

Note that $\textbf{Y}_{p,j}$ is quantized by the low-precision ADCs before being processed for channel estimation.
According to \eqref{ADC} and \eqref{F}, the received pilot symbols after the ADC quantization can be expressed as
\begin{equation}
\begin{aligned}
\label{ypq}
\textbf{Y}_{qp,j}&= \mathcal{Q}_{AD}\left(\textbf{Y}_{p,j} \right)\\
&=(1\!-\!\rho_{AD})\sqrt{ P_p} \sum_{l=1}^L\bar{\textbf{H}}_{jl} \bm{\Psi}^T+(1\!-\!\rho_{AD})\textbf{n}_{p,j}+\textbf{n}_{qp,j},
\end{aligned}
\end{equation}
where $\textbf{n}_{qp,j}=[\textbf{n}_{qp,j1},\textbf{n}_{qp,j2},...,\textbf{n}_{qp,j\tau}]$ denotes the quantization noise and $\textbf{n}_{qp,ji}\sim\mathcal{CN}(\mathbf{0},\sigma_{qp,j}^2\textbf{I}_N)~(i=1,2,...,\tau)$.
Applying a popular discrete Fourier transform matrix $\bm{\Psi}$, by substituting \eqref{C_nq}, \eqref{H_B}, and \eqref{Heq}, the quantized noise power equals \cite{CE2}
\begin{align}
\label{C_npq}
\sigma_{pq,j}^2=\rho_{AD}(1-\rho_{AD})\left(\sigma_n^2+\frac{P_p}{\tau}\sum_{l=1}^L\sum_{k=1}^K \beta_{jlk} |c_{jlk}|^2 \right).
\end{align}

After the ADC quantization, the minimum mean-square-error (MMSE) estimator in \cite{MMSE} is used. The channel estimate can be expressed as
\begin{equation}
\begin{aligned}
\label{H_est}
&\hat{\textbf{H}}_{jj}=\frac{1}{(1-\rho_{AD})\sqrt{ P_p}}\textbf{Y}_{qp,j}\bm{\Psi}^* \textbf{G}_j\\
&\!=\!\!\left(\!\bar{\textbf{H}}_{jj}\!+\!
\underbrace{\!\sum_{l\neq j}\bar{\textbf{H}}_{jl}\!+\!\frac{1}{\!\sqrt{\! P_p}}\textbf{n}_{p,j}\bm{\Psi}^*\!+\!\frac{1}{(1\!-\!\rho_{AD})\!\sqrt{\!P_p}}\textbf{n}_{pq,j}\bm{\Psi}\!^*\!
}_{\textbf{E}_j}
\!\right)\!\textbf{G}_j,
\end{aligned}
\end{equation}
where $\textbf{G}_j$ is the estimation matrix and $\textbf{E}_j$ is the channel estimation error matrix denoted as $\textbf{E}_j=[\textbf{e}_{j1},\textbf{e}_{j2},...,\textbf{e}_{jK}]$.
By using \eqref{Heq}, we have
\begin{align}
\textbf{e}_{jk}=&\sum_{l\neq j}\beta_{jlk}^{\frac{1}{2}}c_{jlk}\textbf{h}_{B,jlk}+\frac{1}{\sqrt{ P_p}}\textbf{n}_{p,j}\bm{\phi}_k^*\nonumber\\
&~~~~~~~~~~~~~~~+\frac{1}{(1-\rho_{AD})\sqrt{ P_p}}\textbf{n}_{pq,j}\bm{\phi}_k^*.
\label{e}
\end{align}
To obtain $\textbf{G}_j$ via MSE minimization, we utilize the asymptotical orthogonality of $\bar{\textbf{H}}_{jl}$ for large $N$, which is presented \emph{Lemma~\ref{lemma_ortho}} in Appendix~D.
Further from \eqref{H_B} and \eqref{Heq}, $\textbf{G}_j$ is directly derived as
\begin{equation}
\begin{aligned}
\label{G}
\textbf{G}_j
&\triangleq\textbf{B}_{jj}\textbf{C}_{jj}^H\textbf{C}_{jj}
\left[\sum_{l=1}^L\textbf{B}_{jl}\textbf{C}_{jl}^H\textbf{C}_{jl}+\mu_j\textbf{I}_K\right]^{-1},
\end{aligned}
\end{equation}
where we define that $\textbf{C}_{jl}\triangleq\textrm{diag}\{c_{jl1},c_{jl2},...,c_{jlK}\}$, $\textbf{B}_{jl}\triangleq\textrm{diag}\{\beta_{jl1},\beta_{jl2},...,\beta_{jlK}\}$ for $j,l=1,2,...,L$, and
\begin{equation}
\begin{aligned}
\label{mu}
\mu_j\triangleq\frac{\sigma_n^2}{ P_p}+\frac{\sigma_{pq,j}^2}{(1-\rho_{AD})^2 P_p}.
\end{aligned}
\end{equation}

\section{Uplink Achievable Rate}
In this section, we are ready to analyze the uplink achievable rate with low-precision ADC quantization and the above channel estimation. We also derive a tight lower bound for the achievable rate, which provides more insights.

\subsection{Ergodic Achievable Rate Analysis}
Let us begin with the expression of the uplink received signal with estimated CSI.
Substituting the estimated channel matrix in \eqref{H_est} into \eqref{y_ul_2}, the detected received vector is expressed as
\begin{align}
\textbf{y}_{j}=&(1\!-\!\rho_{AD})\!\sqrt{P_t}\textbf{G}_j^H(\bar{\textbf{H}}_{jj}^H+\textbf{E}_{j}^H)\sum_{l=1}^L\bar{\textbf{H}}_{jl}\textbf{x}_l\nonumber\\
&+(1-\rho_{AD})\textbf{G}_j^H(\bar{\textbf{H}}_{jj}^H+\textbf{E}_{j}^H)\textbf{n}_j+\textbf{G}_j^H(\bar{\textbf{H}}_{jj}^H+\textbf{E}_{j}^H)\textbf{n}_{q,j}.
\label{yd_MRC}
\end{align}
For homogeneous users, without loss of generality, we focus on the detected signal of user $k$, i.e., $y_{jk}$.
From \eqref{yd_MRC} and by substituting \eqref{Heq}, the detected signal of user $k$ is
\begin{align}
&\frac{y_{jk}}{g_{jk}}=
\nonumber\\
&\underbrace{
(\!1\!-\!\rho_{AD}\!)\sqrt{P_t}\left(\!\beta_{jjk}^\frac{1}{2}c_{jjk}^*\textbf{h}_{B,jjk}^H\!+\!\textbf{e}_{jk}^H\!\right)
\sum_{l=1}^L  \! \sum_{i=1}^K\!\beta_{jli}^\frac{1}{2}c_{jli}\textbf{h}_{B,jli}x_{li}
}_{S_{r,jk}}\nonumber \\
&+\underbrace{
(1-\rho_{AD})\left(\beta_{jjk}^\frac{1}{2}c_{jjk}^*\textbf{h}_{B,jjk}^H+\textbf{e}_{jk}^H\right)\textbf{n}_j
}_{I_{n,jk}} \nonumber\\
&+\underbrace{
\left(\beta_{jjk}^\frac{1}{2}c_{jjk}^*\textbf{h}_{B,jjk}^H+\textbf{e}_{jk}^H\right)\textbf{n}_{q,j}
}_{I_{q,jk}},
\label{yd_MRC_k}
\end{align}
where $g_{jk}$ denotes the $k$th diagonal element of $\textbf{G}_j$ and $x_{jk}$ denotes the $k$th element of vector $\textbf{x}_j$.
In \eqref{yd_MRC_k}, $I_{n,jk}$ represents the equivalent thermal noise, $I_{q,jk}$ denotes the quantization noise, and $S_{r,jk}$ represents the received signal at BS $j$ from all the $LK$ users, among which the desired signal term from user $k$ in cell $j$ equals
\begin{align}
S_{d,jk}=(1-\rho_{AD})\sqrt{P_t}\beta_{jjk}|c_{jjk}|^2\textbf{h}_{B,jjk}^H\textbf{h}_{B,jjk}x_{jk}.
\label{Sd0}
\end{align}

Note that the common scaler, $g_{jk}$, in the left hand side of \eqref{yd_MRC_k} does not affect the evaluation of the received SIQNR.
Therefore, we can drop out $g_{jk}$ and remove the subscript $jk$ for notational brevity.
Using \eqref{Sd0} and \eqref{H_B}, the desired signal power can be expressed as
\begin{align}
S&=\mathbb{E}_{\textbf{x}_j} \{ |S_d|^2  \}\nonumber\\
&=(1-\rho_{AD})^2P_t\beta_{jjk}^2|c_{jjk}|^4N^2.
\label{Sd1}
\end{align}
From \eqref{yd_MRC_k}, we can get the power of interferences and noises as
\begin{align}
I&=\mathbb{E}_{\textbf{x}_1,\textbf{x}_2,...,\textbf{x}_L} \left\{ |S_r+I_n+I_q|^2-|S_d|^2\right\}\nonumber\\
&\overset{(a)}=\mathbb{E}_{\textbf{x}_1,\textbf{x}_2,...,\textbf{x}_L} \left\{ |I_n|^2+|I_q|^2+|S_r|^2-|S_d|^2\right\}\nonumber \\
&\overset{(b)}=\mathbb{E}\{|I_n|^2\}+\mathbb{E}\{|I_q|^2\}+\mathbb{E}\{|S_r|^2\}-S,
\label{I}
\end{align}
where $(a)$ comes from the fact that both channel and quantization noises are uncorrelated with the received signal and $(b)$ utilizes \eqref{Sd1}.
Detailed derivations of the first three terms in \eqref{I} are given in Appendix~A.
Thus far, the SIQNR can be expressed as
\begin{equation}
\begin{aligned}
\label{gamma}
\gamma=\frac{S}{I}.
\end{aligned}
\end{equation}

By applying the assumption of the worst-case Gaussian interference, the ergodic achievable rate of each user can be evaluated as follows
\begin{equation}
\begin{aligned}
\label{rate1}
R=\mathbb{E}\left\{\log(1+\gamma)\right\}.
\end{aligned}
\end{equation}
In most literature on massive MIMO, a concise closed-form expression of $R$ can be further achieved by applying the law of large numbers to the expression of $\gamma$.
The effectiveness relies on the assumption that the dimension of the channel matrix tends large and all the channel coefficients contain a large amount of independent, and possibly identically, random components.
Here as observed in \eqref{In}, \eqref{Iq}, and \eqref{Sr} in Appendix~A, the terms involving the channel coefficients do not tend to an asymptotically deterministic value even with large $N$.
This is because the mmWave MIMO channel is sparse in general.
The sparsity makes the channel matrix, $\textbf{H}_{jlk}$ in \eqref{H}, to have only few terms and the law of large number becomes invalid.
In particular, even with an infinitely large antenna number $N$, $\textbf{H}_{jlk}$ consistently contains only two terms coming from the random angles $\varphi_{jlk}$ and $\theta_{jlk}$ in $\textbf{h}_{U,jlk}$ and $\textbf{h}_{B,jlk}$, respectively.
On the other hand, the AoA estimation error lies on the exponent term in the design of $\textbf{w}_{lk}$ and thus affects the value of $c_{jlk}$ in \eqref{c} highly nonlinearly.
Therefore, the analog beamforming gain, $|c_{jlk}|$, in \eqref{In}, \eqref{Iq}, and \eqref{Sr} is also hard to express in closed form.
Consequently, a direct analysis on \eqref{rate1} is difficult.

\subsection{Lower Rate Bound}
Since the expression of the achievable rate in \eqref{rate1} is complicated, especially the expression of $I$ in \eqref{I}, we derive a tight lower bound for the rate.
Assuming that long-term uplink power control is conducted to compensate for the large-scale fadings of different users in the same cell, the large-scale fading within each cell can be considered identical.
For simplicity, assume that the attenuations between different cells remain the same and we have
\begin{equation}
\label{beta}
\beta_{jlk}=\left\{
\begin{aligned}
&1,~~~~~~~j=l,\\
&\beta, ~~~~~~ j\neq l,
\end{aligned}
\right.
\end{equation}
where $0<\beta<1$.
We have the following theorem on the lower bound for the ergodic achievable rate.

\begin{theorem}
\label{theorem_rate}
A lower bound for the ergodic uplink rate in \eqref{rate1} is given by
\begin{equation}
\label{R_l}
R_{LB}=\log\left[1+\frac{(1-\rho_{AD})^2P_tN^2}{P_u+P_c+P_n+P_q+P_e         }\right],
\end{equation}
where $P_u$ and $P_c$ represent the inter-user and inter-cell interferences, respectively.
$P_n$ is the AWGN and $P_q$ denotes the interference caused by ADC quantization.
$P_e$ represents the interference due to channel estimation error.
They are expressed as:
\begin{equation}
\label{Pu}
P_u=(1-\rho_{AD})^2P_t (K-1) M c^{-2}\eta_2 ,~~~~~~~~~~~~~~~~~~~~~~~~~~~~~~
\end{equation}
\begin{equation}
\label{Pc}
P_c=(1-\rho_{AD})^2 P_t (L-1)K \beta Mc^{-2}\eta_2,~~~~~~~~~~~~~~~~~~~~~~~~~~~~
\end{equation}
\begin{equation}
\begin{aligned}
\label{Pn}
P_n\!=&(1-\rho_{AD})^2\sigma_n^2c^{-4}
\left[ N c^2+N\mu \right.\\
&\left.\!+(\!L\!-\!1\!)N\beta M\!\!+\!\! (\!L\!-\!1\!)(\!L\!-\!2\!)\beta M \eta_1\!+\! 2(\!L\!-\!1\!)\beta^{\frac{1}{2}}cM^{\frac{1}{2}}\eta_1 \!\right]\!\!,
\end{aligned}
\end{equation}
\begin{equation}
\begin{aligned}
\label{Pq}
P_q=&\rho_{AD}(1-\rho_{AD})(\sigma_n^2+\lambda P_t)c^{-4}
\left[ N c^2+N\mu \right.\\
&\left.\!+(\!L\!-\!1\!)N\beta M\!\!+\!\! (\!L\!-\!1\!)(\!L\!-\!2\!)\beta M \eta_1\!+\! 2(\!L\!-\!1\!)\beta^{\frac{1}{2}}cM^{\frac{1}{2}}\eta_1 \!\right]\!\!,
\end{aligned}
\end{equation}
\begin{align}
P_e=&(1-\rho_{AD})^2P_t c^{-4}
\left[N\lambda\mu+(L-1)N^2\beta^2M^2 \right.\nonumber\\
&+2(L-1)(L-2)N\beta^2M^2\eta_1 \nonumber\\
&+2(L-1)N\left(\beta^{\frac{1}{2}}c^3M^{\frac{1}{2}}
+\beta^{\frac{3}{2}}cM^{\frac{3}{2}} \right)\eta_1 \nonumber\\
&+(L-1)K\beta M^2\eta_2+(L-1)(LK-K-1)\beta^2M^2\eta_2\nonumber\\
&+(L-1)(L-2)K\beta M^2\eta_3\nonumber\\
&+(L-1)(L-2)(LK-K-2)\beta^2 M^2\eta_3\nonumber\\
&+2(L-1)(K-1)\beta^{\frac{1}{2}}cM^{\frac{3}{2}}\eta_3\nonumber\\
&\left.+2(L-1)(LK-K-1)\beta^{\frac{3}{2}}cM^{\frac{3}{2}}\eta_3 \right].
\label{Pe}
\end{align}
\end{theorem}

\begin{proof}
See Appendix~E.
\end{proof}

Due to the effect of pilot contamination, the uplink rate converges to a constant with the antenna number increasing to infinity, i.e., $N\rightarrow\infty$.
From \eqref{R_l}, we have
\begin{equation}
\begin{aligned}
\label{R_l_infty}
R_{LB}\rightarrow\log\left[1+\frac{c^4}{(L-1)\beta^2M^2}\right]
,
\end{aligned}
\end{equation}
where we utilize the facts that 
\begin{align}
\frac{\eta_1}{N}=\frac{1}{N}+\frac{\ln N+a}{\pi^2N} \rightarrow 0,
\label{Eta1}
\end{align}
\begin{align}
\frac{\eta_2}{N^2}=\frac{1}{N}-\frac{2}{\pi^2}\left(\frac{1}{N}-\frac{1}{N^2}\right)+\frac{2}{\pi^2}\left(\frac{\ln N}{N} +\frac{a}{N}\right) \rightarrow 0,
\end{align}
and
\begin{align}
\frac{\eta_3}{N^2}&=\frac{\eta_1}{N^2}+\frac{2}{N^2}\sum\limits_{m=1}^{N-1}\sum\limits_{n=0}^{N-m-1} J_0(m\pi)J_0(n\pi)J_0((n+m)\pi)
\nonumber\\
&\overset{(a)}<\frac{\eta_1}{N^2}+\frac{2}{N^2}\sum\limits_{m=1}^{N-1}\sum\limits_{n=0}^{N-m-1} J_0^2(m\pi)
\nonumber\\
&=\frac{\eta_1}{N^2}+\frac{2}{N^2}\sum\limits_{m=1}^{N-1} (N-m) J_0^2(m\pi)
\nonumber\\
&<\frac{\eta_1}{N^2}+\frac{2(N-1)}{N^2}\sum\limits_{m=1}^{N-1}  J_0^2(m\pi)
\nonumber\\
&\overset{(b)}=\frac{\eta_1}{N^2}+\frac{2(N-1)}{N^2\pi^2}(\ln N +a)
\nonumber\\
&\overset{(c)}\rightarrow 0,
\label{Eta3_3}
\end{align}
where $(a)$ use the property of the Bessel function that $J_0(0)=1$, $J_0(n\pi)<1$ and $J_0((n+m)\pi)<J_0(m\pi)$ for $n>0$ \cite{handbook},
$(b)$ comes from $\sum\limits_{m=1}^{N-1}  J_0^2(m\pi)\rightarrow \frac{1}{\pi^2}(\ln N+a)$ as indicated in \eqref{eta1},
and $(c)$ utilizes \eqref{Eta1}.
Further considering $\eta_3>0$, the result in \eqref{Eta3_3} implies $\frac{\eta_3}{N^2}\rightarrow 0$.
From \eqref{R_l_infty}, the desired signal is interfered by signals from other $L-1$ cells with large-scaling fading $\beta$ due to pilot reuse.
$c^4$ is a lower bound for the analog beamforming gain at the user side while $M^2$ represents an upper bound for beamforming gain from other cells.
Note that the asymptotic SIQNR is not affected by the ADC distortion factor $\rho_{AD}$.
It is because that the dominating interference caused by pilot contamination is quantized as well as the desired signal.

\section{Rate Analysis for Single-Cell Scenario}
Pilot contamination suppression has been widely investigated in literature, e.g., in \cite{mul_cell}.
This section then pays attention to the performance of a single-cell network where pilot contamination is temporarily assumed well suppressed.
By setting $L=1$ in \emph{Theorem~\ref{theorem_rate}}, we obtain the lower bound for the achievable data rate in a single-cell network as in \eqref{R_l_sig} at the top of the next page, where $\gamma_t\triangleq\frac{P_t}{\sigma_n^2}$ and $\gamma_p\triangleq\frac{P_p}{\sigma_n^2}$ are the uplink data and pilot SNRs, respectively.
Obviously, BS antenna number $N$, user antenna number $M$, ADC distortion factor $\rho_{AD}$, data SNR $\gamma_t$, and pilot SNR $\gamma_p$ contribute differently to the achievable rate.
In addition, the rate decreases with increasing $K$ since more users cause more pronounced multiuser interference.
Due to the large frequency bandwidth in mmWave communications and the use of massive MIMO, low SNR is able to provide satisfactory data transmission rate \cite{low_SNR1} \cite{low_SNR2}.
In the following, we therefore focus on low SNR scenarios, which is of common interest in mmWave massive MIMO applications.

\newcounter{TempEqCnt}
\setcounter{TempEqCnt}{\value{equation}}
\setcounter{equation}{42}
\begin{figure*}[tb]
 \begin{equation}
\label{R_l_sig}
R_{LB,s}\!=\!\log\!\left[\!1\!+\!\frac{(1-\rho_{AD})^2N^2}{ \frac{c^{-4}N}{\gamma_t \gamma_p}\!+\!c^{-2}N\left(\!\frac{1-\rho_{AD}+\rho_{AD}c^{-2}\lambda/\tau}{\gamma_t}\!+\!\frac{c^{-2}\lambda}{\gamma_p}\!\right) \!+\!\left(\!1\!-\!\rho_{AD}\!+\!\frac{c^{-2}\lambda}{\tau}\!\right)\rho_{AD}Nc^{-2}\lambda\!+\!      (\!1\!-\!\rho_{AD})^2M(K\!-\!1) c^{-2}\eta_2}\!\right].
\end{equation}
\rule[0pt]{18.1cm}{0.05em}
\end{figure*}
\setcounter{equation}{\value{TempEqCnt}}
\setcounter{equation}{43}

\subsection{Imperfect CSI with Low Pilot SNR}
First, we consider the case with low data and pilot SNRs, i.e., $\gamma_t\ll1$ and $\gamma_p\ll1$.
Under this condition, the lower bound in \eqref{R_l_sig} can be further simplified as
\begin{align}
R_{LB,1}&\overset{(a)}=\log\left[1+(1-\rho_{AD})^2NM^2\textrm{sinc}^4 \left(\frac{M}{2}\pi \zeta\right)\gamma_p \gamma_t\right]\nonumber\\
&\overset{(b)}\approx\log\left[1+(1-\rho_{AD})^2NM^2\gamma_p\gamma_t\right]\nonumber\\
&\triangleq\log\left(1+\xi_1\gamma_t\right),
\label{Rl_E1}
\end{align}
where $(a)$ follows by substituting the expression of $c$ in \emph{Lemma~\ref{lemma_c}} and applying the assumption that $\gamma_t\ll1$ and $\gamma_p\ll1$, and $(b)$ comes from the fact that $\textrm{sinc}^4 \left(\frac{M}{2}\pi \zeta\right)\approx 1$ when the analog beamforming interval $\zeta$ is small enough.
It is obvious that the SIQNR is a scaled value of the data SNR $\gamma_t$ by a factor
\begin{equation}
\begin{aligned}
\label{xi_1}
\xi_1\triangleq(1-\rho_{AD})\times NM\times (1-\rho_{AD})M\gamma_p,
\end{aligned}
\end{equation}
where $(1-\rho_{AD})$ represents the SNR attenuation due to the low-precision ADC quantization to the received data signals, and $NM$ represents the beamforming gain at both the BS and user sides.
Factor $(1-\rho_{AD})M\gamma_p$ represents the SNR attenuation due to the channel estimation error.
Specifically, the channel estimation error is mainly caused by AWGN with $\gamma_p\ll1$ and the pilot quantization error from low-precision ADCs, while the analog beamforming at user side improves the estimation accuracy by $M$.


From the above discussion, we have the following important remarks:

1)
From \eqref{Rl_E1}, the achievable rate per user is independent of user number $K$.
This is because that the channel estimation error is mainly caused by AWGN under the assumption of $\gamma_p\ll1$, which overwhelms the effect of multiuser pilot interference.
When transmitting data with $\gamma_t\ll1$, the inter-user interference is negligibly small compared to the thermal noise and the interference caused by imperfect CSI.
In this condition, a large user number hardly degrades the achievable rate of each user.

2)
Expression \eqref{Rl_E1} explicitly characterizing the relationship between increasing the antenna number and the reduction in $\gamma_p$ and $\gamma_t$.
\emph{In particular, a $3$ dB reduction in data or pilot SNR needs doubling the BS antennas, or alternatively increasing user antennas by $\sqrt{2}$ times, in order to maintain the same rate at a low SNR.}
Therefore, increasing antenna number at the user side is more efficient than that at the BS.
However, in practice, the number of antennas at the user side is more tightly restricted by the size of terminals than that at BS.

3)
For fixed $\gamma_t$, $R_{LB,1}$ remains the same if $\xi_1$ in \eqref{xi_1} keeps as a constant.
\emph{
More antennas or higher pilot power can compensate for the rate loss caused by low-precision ADCs.}
According to typical values of $\rho_{AD}$ \cite{rho}, the BS needs $2.5$ times receiving antennas when ADC resolution $b$ decreases from $5$ to $1$, in order to maintain the same rate.
In particular at a low SNR, employing $N=32~(64,~96)$ antennas with 5-bit ADCs at the BS achieves the same rate as using $N=80~(160,~240)$ antennas with 1-bit ADCs, which is also verified by numerical results in Section~\uppercase\expandafter{\romannumeral6}.B.

\subsection{Imperfect CSI with ADC Quantization Error}
In order to improve the accuracy of channel estimation, the pilot power may be set higher than the data transmit power in applications.
Here we assume that $\gamma_p\gg1$ to clearly see the impact of the low-precision ADCs on channel estimation.
In this case, the lower bound in \eqref{R_l_sig} approximately equals
\begin{align}
R_{LB,2}&\overset{(a)}=\log\left[1+\frac{(1-\rho_{AD})^2Nc^2\gamma_t}{1-\rho_{AD}+\rho_{AD}c^{-2}\lambda/\tau}\right]\nonumber\\
&\overset{(b)}\approx\log\left[1+\frac{(1-\rho_{AD})^2}{1-\rho_{AD}+\frac{K}{\tau}\rho_{AD}}NM\gamma_t\right]\nonumber\\
&\triangleq\log\left(1+\xi_2\gamma_t\right),
\label{Rl_H1}
\end{align}
where $(a)$ comes from the assumption that $\gamma_t\ll1$ and $\gamma_p\gg1$.
$(b)$ substitutes the definitions of $c$ and $\lambda$ in \eqref{C3} and \eqref{lambda}, respectively, and uses the approximation $\textrm{sinc}^2 \left(\frac{M}{2}\pi \zeta\right)\approx 1$ for small $\zeta$. The scaling factor is defined as
\begin{align}
\xi_2&\triangleq(1-\rho_{AD}) \times NM \times \frac{1-\rho_{AD}}{1-\rho_{AD}+\frac{K}{\tau}\rho_{AD}}\nonumber\\
&= \frac{1}{M\gamma_p\left(1-\rho_{AD}+\frac{K}{\tau}\rho_{AD}\right)} \xi_1,
\label{xi_2}
\end{align}
which shares some similarities as in \eqref{xi_1}. The factor $(1-\rho_{AD})$ represents the SNR attenuation caused by low-precision ADCs, and $NM$ represents the array gain obtained by beamforming.
The difference between $\xi_2$ and $\xi_1$ is the factor $\frac{1-\rho_{AD}}{1-\rho_{AD}+\frac{K}{\tau}\rho_{AD}}$, representing the SNR attenuation due to different channel estimation qualities.

Based on the above result, we have the following remarks:

1)
Comparing $\xi_2$ in \eqref{xi_2} with $\xi_1$ in \eqref{xi_1}, the difference lies in the last multiplicative term because the dominating factors for the imperfect CSI are different.
In channel estimation, multiuser interference exists because the received pilot signals are quantized by low-precision ADCs.
This quantization operation, to some extent, breaks the orthogonality among pilots from different users in $\bm{\Psi}$.
Under the assumption of $\gamma_p\gg1$, the channel estimation error due to ADC quantization, instead of AWGN, becomes dominating.
On one hand, the channel estimation error decreases with $\tau$ because longer pilot improves the channel estimation accuracy.
On the other hand, a larger $K$ yields more channel estimation error and consequently leads to a lower rate.
For a specific choice of $\tau=K$, the term $\frac{1-\rho_{AD}}{1-\rho_{AD}+\frac{K}{\tau}\rho_{AD}}$ reduces to $1-\rho_{AD}$, which becomes independent of $K$.
This is because the channel estimation error caused by ADC quantization no longer relies on $K$ when the pilot length $\tau$ changes with $K$ simultaneously.

2)
It is obvious that $R_{LB,2}$ remains the same if $\xi_2\gamma_t$ keeps a constant.
On one hand, more antennas can compensate for the reduction in rate with decreased transmit power.
For example, doubling BS antennas, or user antennas, can achieve the same rate with 3 dB lower transmit power.
On the other hand, \emph{the numbers of BS and user antennas can compensate for each other under the constraint that $NM$ remains a constant}.

3)
In order to obtain the effective CSI, the required pilot length is under the constraint that $\tau\geq K$.
Then, we have $\xi_2\geq (1-\rho_{AD})^2 NM$ according to \eqref{xi_2}. As for $\xi_1$ in \eqref{xi_1}, we have $\xi_1\leq (1-\rho_{AD})^2 NM$ requiring $M\gamma_p\leq 1$ under the assumption that $\gamma_p \ll 1$.
Thus, we have
\begin{equation}
\begin{aligned}
\label{xi}
\xi_2\geq \xi_1,
\end{aligned}
\end{equation}
which is reasonable since the high pilot SNR always provides better rate performance than the low SNR case.
Note that all the above insights observed from a single-cell condition are also valid for multi-cell scenarios, which is verified in the next section.

Please note that the observations and derivation results in this paper are based on a common assumption of perfect synchronization in frequency domain.
In general, this can be achieved by using existing synchronization techniques \cite{CFO_MIMO1}-\cite{CFO_MIMO3}.
While recent works \cite{CFO1}-\cite{CFO3} have shown that frequency synchronization is a challenging issue for implementation in massive MIMO due to prohibitively increasing complexity with a large antenna number.
In \cite{CFO1}, a constant envelope pilot signal based carrier frequency offset (CFO) estimation has been proposed for massive MIMO systems.
A blind frequency synchronization method for multiuser massive MIMO uplink transmission has been presented in \cite{CFO2}.
By exploiting the angle information of users, a new frequency synchronization scheme has been designed in \cite{CFO3}.
These recently proposed synchronization methods can be applied in mmWave massive MIMO networks to guarantee that our assumption makes sense.

\section{Simulation Results}
In this section, we verify the derived lower rate bound in \eqref{R_l} by numerical examples and test the effect of various system parameters on the rate performance.
In the following, the inter-cell distortion factor $\beta$ is set to be $0.1$ for moderate distance between adjacent cells \cite{mul_cell}.
The phase shifter resolution is set to $B=6$, which has been shown accurate enough in practice \cite{CE1}.
In order to reduce the pilot overhead, we set $\tau=K$ under the constraint of $\tau\geq K$, unless otherwise specified.

\subsection{Lower Rate Bound Verifications}
\begin{figure}[tb]
\centering\includegraphics[width=0.51\textwidth,bb=20 220 580 620]{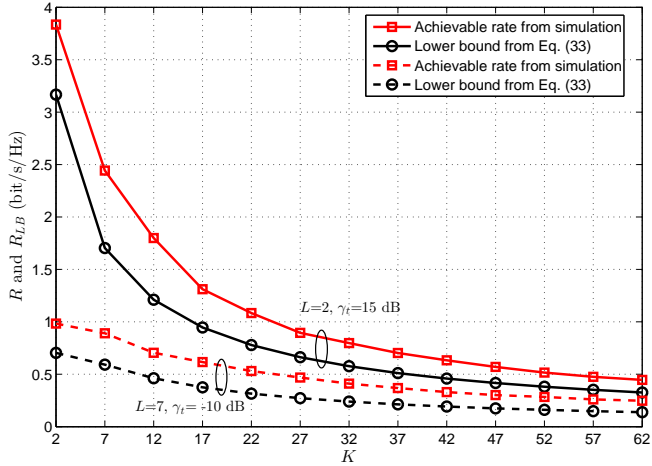}
\caption{The lower bound for uplink achievable rate versus the number of users.}
\label{fig_bound_ul}
\end{figure}

Fig. \ref{fig_bound_ul} compares the uplink achievable rate and the lower bound in \eqref{R_l}.
We set $N=64$, $M=2$, and $\gamma_p=\tau \gamma_t$. 1-bit ADCs are adopted.
From this figure, the achievable rate first decreases and then converges to a constant with increasing $K$.
This is because the interference caused by channel estimation error dominates, overwhelming the inter-user interference even with a large user number.
In general, our derived bound is tight with user number $K$ ranging from $2$ to $62$.
Moreover, the bound tends tighter with increasing $K$ due to the use of Jensen's inequality.

\subsection{Imperfect CSI with Low Pilot SNR}

\begin{figure}[tb]
\centering\includegraphics[width=0.51\textwidth,bb=20 220 580 620]{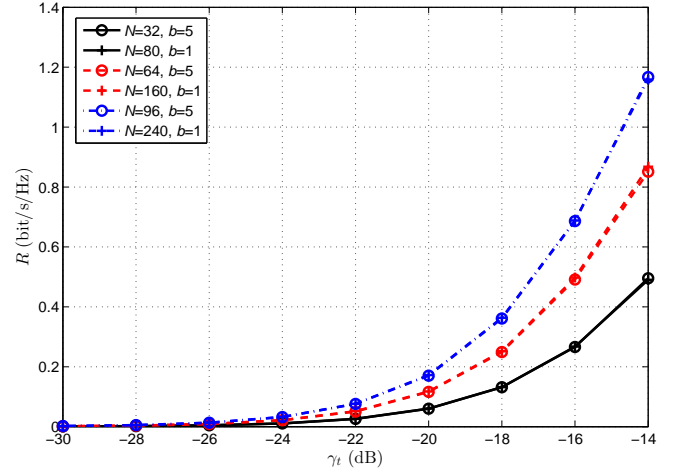}
\caption{Achievable rate versus data SNR with various BS antenna numbers and ADC precisions.}
\label{fig_b}
\end{figure}

\begin{figure}[tb]
\centering\includegraphics[width=0.51\textwidth,bb=20 220 580 620]{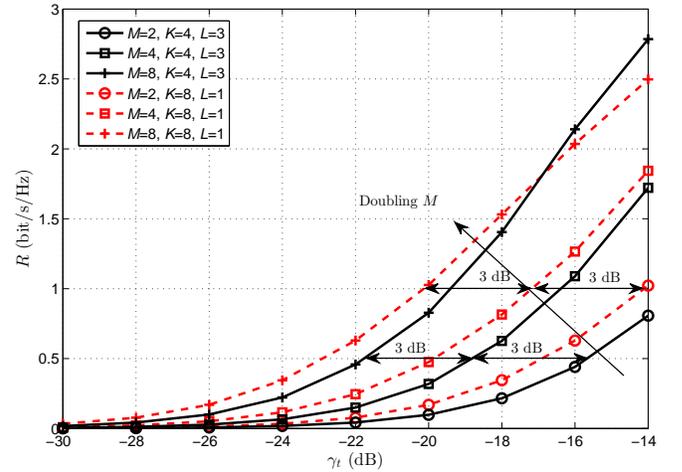}
\caption{Achievable rate versus data SNR with various user antenna numbers.}
\label{fig_M}
\end{figure}

We show the achievable rates from numerical results under the assumptions that $\gamma_t\ll1$ and $\gamma_p\ll1$.
Fig. \ref{fig_b} displays the achievable rates with different ADC precisions and BS antenna numbers. We set $L=7$, $K=4$, $M=2$, and $\gamma_p=\tau \gamma_t$.
We observe that $2.5$ times antennas at BS achieves almost the same rate when ADC precision $b$ decreases from $5$ to $1$.
This implies that more receiving antennas can effectively compensate for low-precision ADC quantization distortion.

Fig. \ref{fig_M} shows the achievable rate versus data SNR $\gamma_t$ with user antenna numbers $M=2, 4,$ and $8$, using 1-bit ADCs. We set $\gamma_p=\tau \gamma_t$ and $N=128$.
For the single-cell condition with $L=1$, doubling $M$ can trade for a reduction in both $\gamma_t$ and $\gamma_p$ by $3$ dB, as we have mentioned before.
This implies that adding antennas at user side can compensate for the SNR reduction.
While for the multi-cell case with $L=3$, similar observations can be obtained.

\subsection{Imperfect CSI with ADC Quantization Error}


\begin{figure}[tb]
\centering\includegraphics[width=0.51\textwidth,bb=20 220 580 620]{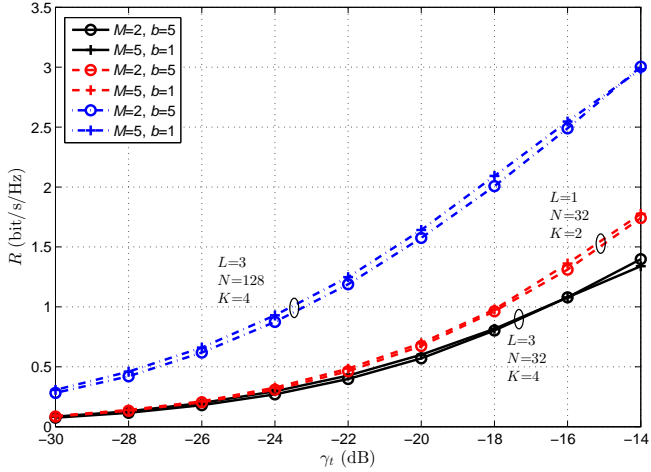}
\caption{Achievable rate versus data SNR with various user antenna numbers and ADC precisions. }
\label{fig_M_b}
\end{figure}

\begin{figure}[tb]
\centering\includegraphics[width=0.51\textwidth,bb=20 220 580 620]{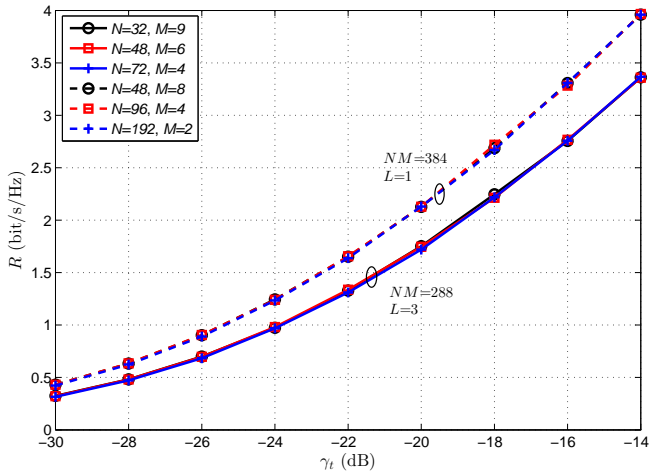}
\caption{Achievable rate with keeping the value of $NM$ as a constant. }
\label{fig_NM}
\end{figure}

\begin{figure}[tb]
\centering\includegraphics[width=0.51\textwidth,bb=20 220 580 620]{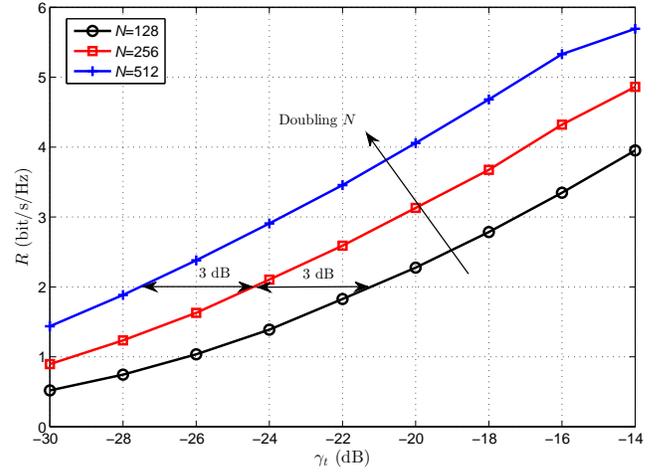}
\caption{Achievable rate versus data SNR with various BS antenna numbers.}
\label{fig_N}
\end{figure}

In the following, we show the achievable rates by numerical simulations under a low data SNR but with high pilot power, i.e., $\gamma_t\ll1, \gamma_p\gg1$.
Fig. \ref{fig_M_b} shows the achievable rate with various user antenna numbers and ADC precisions, under both single- and multi-cell conditions.
The pilot SNR is set as $\gamma_p=10$~dB.
It can be observed that 2.5 times more user antennas can approximately provide the same rate with the number of the ADC quantized bits decreasing from 5 to 1.
This is because under these two scenarios, the SNR scaling factor $\xi_2$ remains a constant with $\tau=K$ in \eqref{xi_2}.
It implies that adding antennas at the user side can also compensate for the low-precision quantization distortions at the BS.

Fig. \ref{fig_NM} displays the achievable rate with $NM$ maintaining as a constant. The parameters are set as $K=2$ and $\gamma_p=10$~dB. 3-bit ADCs are exploited.
For the single-cell case, i.e., $L=1$, we set $NM=384$ while for the multi-cell case with $L=3$, we set $NM=288$.
From this figure, the rates remain the same when keeping $NM$ as a constant under low SNRs.
It implies that adding antennas at the user side can compensate for the lack of antennas at BS side, and vice versa.

Fig. \ref{fig_N} shows the achievable rate versus data SNR $\gamma_t$ with various BS antenna numbers $N=128, 256, 512$, equipping 3-bit ADCs.
We set that $L=3$, $\gamma_p=10$ dB, $K=4$, and $M=4$.
We observe that doubling $N$ can approximately compensate for the rate loss due to a $3$ dB reduction in $\gamma_t$, as indicated before.
More antennas at BS can compensate for the SNR reduction.

\subsection{Comparison Between Low and High Pilot SNRs}

\begin{figure}[tb]
\centering\includegraphics[width=0.51\textwidth,bb=20 220 580 620]{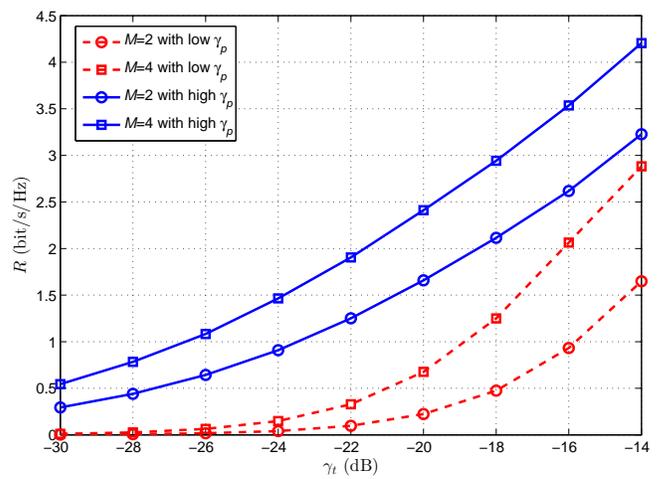}
\caption{Achievable rate versus data SNR for low and high pilot powers.}
\label{fig_comp}
\end{figure}

\begin{figure}[tb]
\centering\includegraphics[width=0.51\textwidth,bb=20 220 580 620]{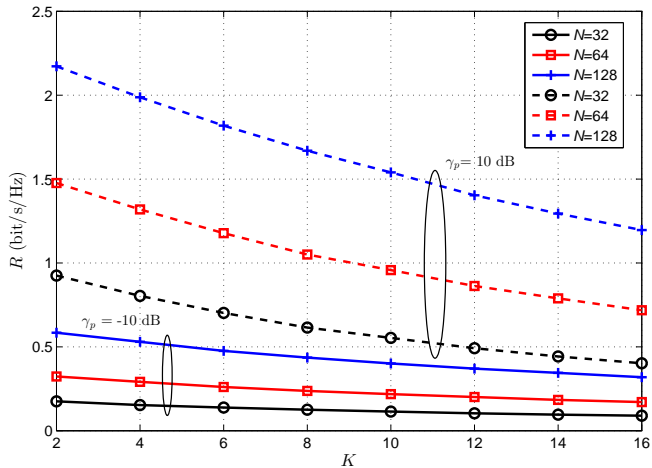}
\caption{Achievable rate versus user number.}
\label{fig_K}
\end{figure}

Fig. 8 compares the achievable rates for low and high pilot powers, i.e., $\gamma_p=\tau \gamma_t$ and $\gamma_p=10$~dB.
We set $L=1$, $K=4$, and $N=128$, and use 3-bit ADCs.
It is observed that the achievable rate of lower pilot power increases more significantly with increasing $M$ than that of higher pilot power.
Under the scenario of small $\gamma_p$, it is therefore more efficient to increase $M$ for performance improvement.

Fig. \ref{fig_K} shows the achievable rate versus user number with a low SNR, i.e., $\gamma_t=-15$ dB, using 3-bit ADCs.
We select $K$ ranging from $2$ to $16$ and choose a fixed $\tau=16$ guaranteeing $\tau\geq K$.
We set $L=3$, $M=2$ and compare the high pilot power scenario, i.e., $\gamma_p=10$ dB, with the low pilot power condition, i.e., $\gamma_p=-10$ dB.
For high pilot SNR, the rates decrease with user number $K$ increasing from $2$ to $16$.
This is because the multiuser interference in channel estimation increases with $K$ due to the low-precision ADC quantization.
While for low pilot SNR, the rates are approximately independent of $K$ since the channel estimation error is mainly caused by channel noise as indicated before.
It implies that adding users will not cause more rate loss with both low data and pilot SNRs.

\section{Conclusion}

In this work, we consider a multi-cell mmWave networks using large antenna arrays. The BS equips low-precision ADCs while each multi-antenna user is driven by a single RF chain.
Considering the ADC quantization distortion and the analog beamforming gain, we analyze the uplink achievable rate with imperfect CSI.
Furthermore, a tight lower bound for the user rate is derived.
Specially, we focus on a single-cell case and find that the received SIQNR can be expressed as a scaling value of the original low SNR.
The scaling factor is proportional to $NM^2$ with low pilot SNR while proportional to $NM$ with high pilot SNR, in which case the channel estimation error is mainly caused by ADC quantization.
The system parameters, including the antenna numbers at both the BS and user sides, the ADC precision, and the data and pilot SNRs, can be adjusted in order to balance the rate performance.
A mixed-ADC architecture under more general channel models, e.g., Rician fading channels, could be studied in future work.

\begin{appendices}
\section{Derivation of terms in \eqref{I}}
In this appendix, we derive the first three expectation terms in \eqref{I} one by one.
From \eqref{H_est}, we first give the expression of the estimated channel from user $k$ in cell $j$ to BS $j$ as
\begin{align}
&\frac{\hat{\textbf{h}}_{jjk}}{g_{jk}}
\nonumber
\\
=&\beta_{jjk}^\frac{1}{2}c_{jjk}\textbf{h}_{B,jjk}+\textbf{e}_{jk}
\nonumber
\\
=&\sum_{l=1}^L \beta_{jlk}^{\frac{1}{2}}c_{jlk}\textbf{h}_{B,jlk}+
\underbrace{
\frac{1}{\!\sqrt{ P_p}}\textbf{n}_{p,j}\bm{\phi}_k^*
\!+\!\frac{1}{(\!1\!-\!\rho_{AD}\!)\!\sqrt{ P_p}}\textbf{n}_{pq,j}\bm{\phi}_k^*
}_{\tilde{\textbf{n}}}
,
\label{CE_1}
\end{align}
where $\tilde{\textbf{n}}\sim\mathcal{CN}(\mathbf{0},\mu_j\textbf{I}_N)$ is defined as an equivalent estimation noise vector with $\mu_j$ defined in \eqref{mu}, including thermal and ADC quantization noise.
Note that pilot contamination exists in \eqref{CE_1} as the channel vectors from users in other cells are also contained in the estimate.

From \eqref{yd_MRC_k}, the interference power caused by channel AWGN can be expressed as
\begin{align}
&\mathbb{E} \left\{ |I_n|^2  \right\}
\nonumber
\\
=&(1-\rho_{AD})^2 \mathbb{E} \left\{
\left|\left(\beta_{jjk}^\frac{1}{2}c_{jjk}^*\textbf{h}_{B,jjk}^H+\textbf{e}_{jk}^H\right)\textbf{n}_j\right|^2
\right\}
\nonumber
\\
\overset{(a)}=&(1-\rho_{AD})^2\mathbb{E} \left\{
\left(\sum_{l=1}^L\beta_{jlk}^{\frac{1}{2}}c_{jlk}^*\textbf{h}_{B,jlk}^H+\tilde{\textbf{n}}^H\right)\textbf{n}_j\textbf{n}_j^H
\right.
\nonumber
\\
&\left.~~~~~~~~~~~~~~~~~~~~~~~~~\times
\left(\sum_{l=1}^L \beta_{jlk}^{\frac{1}{2}}c_{jlk}\textbf{h}_{B,jlk}+ \tilde{\textbf{n}}\right)
\right\}
\nonumber
\\
\overset{(b)}=&(1-\rho_{AD})^2\sigma_n^2\mathbb{E} \left\{
\left(\sum_{l=1}^L\beta_{jlk}^{\frac{1}{2}}c_{jlk}^*\textbf{h}_{B,jlk}^H+ \tilde{\textbf{n}}^H\right)
\right.
\nonumber
\\
&\left.~~~~~~~~~~~~~~~~~~~~~~~~\times
\left(\sum_{l=1}^L \beta_{jlk}^{\frac{1}{2}}c_{jlk}\textbf{h}_{B,jlk}+ \tilde{\textbf{n}}\right)
\right\}
\nonumber
\\
\overset{(c)}=&(\!1\!-\!\rho_{AD}\!)^2\sigma_n^2\!\left(\!\!
N\mu_j\!+\!\!\sum_{l=1}^L \! \sum_{t=1}^L \! \beta_{jlk}^{\frac{1}{2}}\beta_{jtk}^{\frac{1}{2}}c_{jlk}^*c_{jtk} \textbf{h}_{B,jlk}^H\textbf{h}_{B,jtk}
\!\!\right)
\nonumber
\\
\overset{(d)}=&(\!1\!-\!\rho_{AD}\!)^2\sigma_n^2\left(\!
 N\mu_j +N \sum_{l=1}^L \beta_{jlk} |c_{jlk}|^2
\right.
\nonumber
\\
&~~~~~~~~~~~~\left.+\sum_{l=1}^L \sum_{t\neq l} \beta_{jlk}^{\frac{1}{2}}\beta_{jtk}^{\frac{1}{2}}c_{jlk}^*c_{jtk} \textbf{h}_{B,jlk}^H\textbf{h}_{B,jtk}
\!\!\right)
\!,
\label{In}
\end{align}
where $(a)$ uses \eqref{CE_1}, $(b)$ comes from the fact that AWGN $\textbf{n}$ is uncorrelated with the estimated channel vector, $(c)$ utilizes the fact that estimation noise $\tilde{\textbf{n}}$ is uncorrelated with channel vectors, and $(d)$ uses \eqref{H_B}.

Similarly for the interference due to ADC quantization, i.e., $I_q$ in \eqref{yd_MRC_k}, we have
\begin{align}
&\mathbb{E} \left\{ |I_q|^2  \right\}\nonumber\\
=&\mathbb{E} \left\{
\left|\left(\beta_{jjk}^\frac{1}{2}c_{jjk}^*\textbf{h}_{B,jjk}^H+\textbf{e}_{jk}^H\right)\textbf{n}_{q,j}\right|^2
\right\}
\nonumber
\\
\overset{(a)}=&\mathbb{E} \left\{
\left(\sum_{l=1}^L\beta_{jlk}^{\frac{1}{2}}c_{jlk}^*\textbf{h}_{B,jlk}^H+ \tilde{\textbf{n}}^H\right)\textbf{n}_{q,j}\textbf{n}_{q,j}^H
\right.
\nonumber
\\
&\left.~~~~~~~~~~~~~~~~~~~~~~~~\times
\left(\sum_{l=1}^L\beta_{jlk}^{\frac{1}{2}}c_{jlk}\textbf{h}_{B,jlk}+ \tilde{\textbf{n}}\right)
\right\}
\nonumber
\\
\overset{(b)}=&\sigma_{q,j}^2\mathbb{E} \left\{
\left(\sum_{l=1}^L\beta_{jlk}^{\frac{1}{2}}c_{jlk}^*\textbf{h}_{B,jlk}^H+ \tilde{\textbf{n}}^H\right)
\right.
\nonumber
\\
&\left.~~~~~~~~~~~~~~~~~~~~~~~~\times
\left(\sum_{l=1}^L\beta_{jlk}^{\frac{1}{2}}c_{jlk}\textbf{h}_{B,jlk}+ \tilde{\textbf{n}}\right)
\right\}
\nonumber
\\
\overset{(c)}=&\sigma_{q,j}^2\!\left(\!\!
N\mu_j\! +\!\sum_{l=1}^L \sum_{t=1}^L \beta_{jlk}^{\frac{1}{2}}\beta_{jtk}^{\frac{1}{2}}c_{jlk}^*c_{jtk} \textbf{h}_{B,jlk}^H\textbf{h}_{B,jtk}\!\!\right)
\nonumber
\end{align}

~\\
\begin{align}
\overset{(d)}=&\sigma_{q,j}^2\left(
N\mu_j +N \sum_{l=1}^L \beta_{jlk} |c_{jlk}|^2
\right.
\nonumber
\\
&~~~~~~~~~~~\left.+\sum_{l=1}^L \sum_{t\neq l} \beta_{jlk}^{\frac{1}{2}}\beta_{jtk}^{\frac{1}{2}}c_{jlk}^*c_{jtk} \textbf{h}_{B,jlk}^H\textbf{h}_{B,jtk} \right)
,
\label{Iq}
\end{align}
where $(a)-(d)$ follow the same reasons as in deriving \eqref{In}.

As for the received signal, i.e., $S_r$ in \eqref{yd_MRC_k}, the signal power can be derived as follows
\begin{align}
&\mathbb{E} \left\{ |S_r|^2  \right\}
\nonumber
\\
\overset{(a)}=&(1-\rho_{AD})^2P_t
\nonumber
\\
&\times
\mathbb{E}\left\{ \!\left|\left(\beta_{jjk}^\frac{1}{2}c_{jjk}^*\textbf{h}_{B,jjk}^H\!+\!\textbf{e}_{jk}^H\right)  \sum_{l=1}^L  \sum_{i=1}^K\beta_{jli}^\frac{1}{2}c_{jli} \textbf{h}_{B,jli}x_{li} \right|^2\!\right\}
\nonumber
\\
\overset{(b)}=&(1-\rho_{AD})^2P_t\mathbb{E}\left\{
\left(\sum_{l=1}^L \beta_{jlk}^{\frac{1}{2}}c_{jlk}^*\textbf{h}_{B,jlk}^H+ \tilde{\textbf{n}}^H\right)
\right.
\nonumber
\\
\nonumber
&\left.\times\!
\left(\!\sum_{l=1}^L \!\sum_{i=1}^K\!\beta_{jli}|c_{jli}|^2 \textbf{h}_{B,jli}\textbf{h}_{B,jli}^H \!\right)\!\!
\left(\!\sum_{l=1}^L \!\beta_{jlk}^{\frac{1}{2}}c_{jlk}\textbf{h}_{B,jlk}\!+\! \tilde{\textbf{n}}\!\right)
\!\!\right\}
\nonumber
\\
\overset{(c)}=&(1\!-\!\rho_{AD})^2P_t \left(\mu_j N \sum_{l=1}^L \sum_{i=1}^K\beta_{jli}|c_{jli}|^2
\!+\!\sum_{t=1}^L \beta_{jtk}^{\frac{1}{2}}c_{jtk}^*\textbf{h}_{B,jtk}^H
\right.
\nonumber
\\
&\left.\times
\sum_{l=1}^L \sum_{i=1}^K\beta_{jli}|c_{jli}|^2 \textbf{h}_{B,jli}\textbf{h}_{B,jli}^H
\sum_{r=1}^L \beta_{jrk}^{\frac{1}{2}}c_{jrk}\textbf{h}_{B,jrk}
\right)
\nonumber
\\
=&(1-\rho_{AD})^2P_t \left(\mu_j N \sum_{l=1}^L \sum_{i=1}^K\beta_{jli}|c_{jli}|^2
+\sum_{t=1}^L \sum_{l=1}^L \sum_{i=1}^K \sum_{r=1}^L
\right.
\nonumber
\\
&\left.\beta_{jtk}^{\frac{1}{2}} \beta_{jli}  \beta_{jrk}^{\frac{1}{2}}
c_{jtk}^*|c_{jli}|^2 c_{jrk}\textbf{h}_{B,jtk}^H  \textbf{h}_{B,jli}\textbf{h}_{B,jli}^H  \textbf{h}_{B,jrk}
\right)
,
\label{Sr}
\end{align}
where $(a)$ utilizes \eqref{yd_MRC_k}, $(b)$ uses \eqref{CE_1}, and $(c)$ comes from the fact that $\tilde{\textbf{n}}$ is uncorrelated with channel vectors.
Finally, by substituting \eqref{In}, \eqref{Iq}, and \eqref{Sr} into \eqref{I}, the expression of interference and noise power can be directly obtained.

~\\

\section{Lemma 1}
\begin{lemma}
\label{lemma_c}
The analog beamforming gain within the cell $l$, $c_{llk}$, defined in \eqref{c} is bounded as follows
\begin{equation}
\begin{aligned}
\label{C3}
\sqrt{M}\geq|c_{llk}|\geq \sqrt{M}\textrm{sinc} \left(\frac{M}{2}\pi \zeta\right) \triangleq c ,
\end{aligned}
\end{equation}
for $l=1,2,...,L$ and $k=1,2,...,K$.
While for the analog beamforming gain from cell $l$ to cell $j$, $c_{jlk}$, there is the same upper bound as follows
\begin{equation}
\begin{aligned}
\label{C4}
|c_{jlk}|\leq \sqrt{M},
\end{aligned}
\end{equation}
for $j,l=1,2,...,L,~j\neq l$, and $k=1,2,...,K$.
\end{lemma}

\begin{proof}
Firstly, we consider the analog beamforming gain $c_{llk}$.
Without causing misunderstanding, the index $l$ and $k$ are removed for brevity in most places of the following derivations. Substituting \eqref{H_U} and \eqref{wk2} into \eqref{c}, we have
\begin{align}
c_{llk}&=\frac{1}{\sqrt{M}}\sum\limits_{n=0}^{M-1}e^{jn\pi(\cos\varphi-\cos\hat{\varphi})}\nonumber\\
&=\frac{1}{\sqrt{M}}e^{\frac{1}{2}j(M\!-\!1)\pi(\cos\varphi-\cos\hat{\varphi})}\frac{\sin[\frac{1}{2}M\pi(\cos\varphi\!-\!\cos\hat{\varphi})]}{\sin[\frac{1}{2}\pi(\cos\varphi-\cos\hat{\varphi})]}\nonumber\\
&=\frac{1}{\sqrt{M}}e^{\frac{1}{2}j(M-1)\pi(\cos\varphi-\cos(\varphi-\Delta\varphi))}\nonumber\\
& ~~~~~~~~~~~~~~~\times \frac{\sin[\frac{1}{2}M\pi(\cos\varphi\!-\!\cos(\varphi-\Delta\varphi))]}{\sin[\frac{1}{2}\pi(\cos\varphi-\cos(\varphi-\Delta\varphi))]},
\label{C}
\end{align}
where we define $\Delta\varphi=\varphi-\hat{\varphi}$ as the AoA estimation error.
Obviously, the upper bound for $\sqrt{M}$ is established.
As for the lower bound, assume $\varphi\sim$U$[0,\pi]$ and $\hat{\varphi}$ is chosen from the codebook in \eqref{psi} in order to maximize $|c_{llk}|$.
Although the AoA estimation is affected by thermal noise $\textbf{n}^{\textrm{A}}$ as indicated in \eqref{AoA} and \eqref{AoA2}, the power of the term $\tilde{\textbf{w}}^T\textbf{n}^{\textrm{A}}$ has a constant expectation over $\textbf{n}^{\textrm{A}}$. Thus, the impact of $\textbf{n}^{\textrm{A}}$ on AoA estimation is neglectable when averaged over the noise term.
We then make the assumption that the estimation error follows the distribution $\Delta\varphi\sim$U$[-\zeta,\zeta]$, where $\zeta$ is the phase interval of a quantized analog beamformer.

We temporarily focus on the condition that $\Delta\varphi<0$. Assuming that the phase shifter resolution is reasonably high so that $\zeta\leq\frac{2}{M}$ always holds, we have
\begin{align}
0\leq\cos\varphi-\cos(\varphi-\Delta\varphi)\leq-\Delta\varphi\leq\zeta\leq\frac{2}{M},
\end{align}
which comes from the fact
\begin{align}
\frac{\cos\varphi-\cos(\varphi-\Delta\varphi)}{\Delta\varphi}
\geq-1.
\end{align}
Thus, we have $0\leq\frac{1}{2}M\pi(\cos\varphi-\cos(\varphi-\Delta\varphi))\leq\pi$ and $0\leq\frac{1}{2}\pi(\cos\varphi-\cos(\varphi-\Delta\varphi))\leq \frac{\pi}{M} \leq \frac{\pi}{2}$ with $M\geq2$, which implies that both the numerator and denominator of the last term in \eqref{C} are positive.
Using this, we have the magnitude of $c_{llk}$ bounded as
\begin{align}
|c_{llk}|
&=\frac{1}{\sqrt{M}}\frac{\sin[\frac{1}{2}M\pi(\cos\varphi-\cos(\varphi-\Delta\varphi))]}{\sin[\frac{1}{2}\pi(\cos\varphi-\cos(\varphi-\Delta\varphi))]}\nonumber\\
&\overset{(a)}\geq\frac{1}{\sqrt{M}}\frac{\sin[\frac{1}{2}M\pi(\cos\varphi-\cos(\varphi-\Delta\varphi))]}{\frac{1}{2}\pi(\cos\varphi-\cos(\varphi-\Delta\varphi))}\nonumber\\
&=\sqrt{M}\textrm{sinc}\left[\frac{1}{2}M\pi\left(\cos\varphi-\cos(\varphi-\Delta\varphi)\right)\right]\nonumber\\
&\overset{(b)}\geq\sqrt{M}\textrm{sinc} \left(\frac{M}{2}\pi \zeta\right),
\label{C4}
\end{align}
where $(a)$ uses the fact that $\sin(x) \leq x$ for $0\leq x\leq \frac{\pi}{2}$ at the denominator, and $(b)$ follows that $\textrm{sinc}(x)$ is a decreasing function w.r.t. $x\in[0,\pi]$.
As for $\Delta\varphi>0$, the conclusion still holds due to symmetry.

Similarly for $c_{jlk}~(j\neq l)$, the upper bound for $\sqrt{M}$ can be easily established.
\end{proof}

\section{Lemma 2-4}
\begin{lemma}
\label{lemma_eta1}
For the two independent channel vectors $\textbf{h}_{B,jlk}$ and $\textbf{h}_{B,j'l'k'}$, where $(j,l,k)\neq (j',l',k')$, it follows for large $N$ that
\begin{equation}
\begin{aligned}
&\mathbb{E}\{\textbf{h}_{B,jlk}^H\textbf{h}_{B,j'l'k'}\}\rightarrow \eta_1,
\end{aligned}
\end{equation}
where $\eta_1\triangleq1+\frac{1}{\pi^2}(\ln N+a)$ and $a$ is the Euler's constant.
\end{lemma}

\begin{proof}
Using \eqref{H_B}, we have
\begin{align}
&\mathbb{E}\{\textbf{h}_{B,jlk}^H\textbf{h}_{B,j'l'k'}\}\nonumber \\
=&\mathbb{E}\left\{\sum_{n=0}^{N-1}e^{jn\pi\left(\cos\theta_{jlk}\!-\!\cos\theta_{j'l'k'}\right)}\right\}\nonumber \\
\overset{(a)}=&\sum_{n=0}^{N-1}\mathbb{E}\left\{\!e^{jn\pi\cos\theta_{jlk}}\!\right\}\mathbb{E}\left\{\!e^{-jn\pi\cos\theta_{j'l'k'}}\!\right\}\nonumber \\
\overset{(b)}=&\sum_{n=0}^{N-1}J_0^2(n\pi)\nonumber \\
\overset{(c)}\rightarrow &1+\frac{1}{\pi^2}\sum_{n=1}^{N-1}\frac{1}{n}\nonumber \\
\overset{(d)}\rightarrow&1+\frac{1}{\pi^2}(\ln N+a).
\label{eta1}
\end{align}
Firstly, $(a)$ uses the fact that $\theta_{jlk}$ is independent of $\theta_{j'l'k'}$.
Secondly, $(b)$ comes from the equality $\mathbb{E}\left\{e^{jn\pi\cos\theta_{jlk}}\right\}=\mathbb{E}\left\{e^{-jn\pi\cos\theta_{j'l'k'}}\right\}=J_0(n\pi)$.
Take $\mathbb{E}\left\{e^{jn\pi\cos\theta_{jlk}}\right\}$ for instance.
Since $\theta_{jlk}$ follows the uniform distribution U$[0,\pi]$, we have
\begin{align}
&~~~~\mathbb{E}\left\{e^{jn\pi\cos\theta_{jlk}}\right\}\nonumber \\
&=\mathbb{E}\left\{\cos (n\pi\cos\theta_{jlk})+j\sin (n\pi\cos\theta_{jlk})\right\}\nonumber \\
&=\!\frac{1}{\pi}\!\int_0^{\pi}\cos (n\pi\cos\theta_{jlk}) \textrm{d}\theta_{jlk}\!+\!\frac{j}{\pi}\int_0^{\pi}\sin (n\pi\cos\theta_{jlk}) \textrm{d}\theta_{jlk}\nonumber \\
&=J_0(n\pi),
\label{J}
\end{align}
where the last step uses the integral equations \cite[Eqs. (18), (13), pp. 425]{table}, and
$J_{\nu}(\cdot)$ is the ${\nu}$th Bessel function.
Similarly, the equality holds for $\mathbb{E}\left\{e^{-jn\pi\cos\theta_{j'l'k'}}\right\}=J_0(n\pi)$.
Thirdly, $(c)$ comes from the fact that \cite[Eq. 9.2.1]{handbook}
\begin{align}
\label{J2}
J_{\nu}(x)\rightarrow \sqrt{\frac{2}{\pi x}}\cos \left(x-\frac{\nu\pi}{2}-\frac{\pi}{4}\right),
\end{align}
for $|x|\rightarrow \infty$ with $\nu=0$.
Note that the asymptotical equality in \eqref{J2} behaves tight even for small $x$.
Finally, $(d)$ follows by the assumption that $N\rightarrow\infty$ and the definition of the Euler's constant \cite{handbook} as
\begin{equation}
\label{Euler}
a\triangleq\lim_{n\rightarrow\infty} \left[\sum_{k=1}^{k=n-1}\frac{1}{k}-\ln n\right].
\end{equation}
\end{proof}

\begin{lemma}
\label{lemma_eta2}
For the two independent channel vectors $\textbf{h}_{B,jlk}$ and $\textbf{h}_{B,j'l'k'}$, where $(j,l,k)\neq (j',l',k')$, it follows for large $N$ that
\begin{equation}
\begin{aligned}
&\mathbb{E}\left\{|\textbf{h}_{B,jlk}^H\textbf{h}_{B,j'l'k'}|^2\right\}\rightarrow \eta_2,
\end{aligned}
\end{equation}
where $\eta_2\triangleq N-\frac{2}{\pi^2}(N-1)+\frac{2N}{\pi^2}(\ln N +a)$ and $a$ is the Euler's constant.
\end{lemma}

\begin{proof}
According to \eqref{H_B}, $\mathbb{E}\{|\textbf{h}_{B,jlk}^H\textbf{h}_{B,j'l'k'}|^2\}$ can be evaluated as
\begin{align}
&\mathbb{E}\left\{|\textbf{h}_{B,jlk}^H\textbf{h}_{B,j'l'k'}|^2\right\}\nonumber\\
\!=\!&\mathbb{E}\left\{\!\sum_{n=0}^{N-1}e^{jn\pi\!(\!\cos\theta_{jlk}\!-\!\cos\theta_{j'l'k'}\!)\!} \!\times\! \sum_{n=0}^{N-1}e^{-jn\pi\!(\!\cos\theta_{jlk}\!-\!\cos\theta_{j'l'k'}\!)\!}\!\right\}\nonumber\\
=&N+\sum_{n=1}^{N-1}(N-n)\nonumber\\
\!\times\!&\left[\mathbb{E}\left\{e^{jn\pi\!(\!\cos\theta_{jlk}\!-\!\cos\theta_{j'l'k'}\!)\!}\right\}\!+\!\mathbb{E}\left\{e^{-jn\pi\!(\!\cos\theta_{jlk}\!-\!\cos\theta_{j'l'k'}\!)\!}\right\}\right]\nonumber\\
\overset{(a)}=&N+\sum_{n=1}^{N-1}(N-n)\times\left[\mathbb{E}\left\{e^{jn\pi\!\cos\theta_{jlk}\!\!}\right\}\mathbb{E}\left\{e^{-jn\pi\!\cos\theta_{j'l'k'}\!\!}\right\}\right.\nonumber\\
&~~~~~~~~~~~~~~~~~~~~~~~~\left.+\mathbb{E}\left\{e^{-jn\pi\!\cos\theta_{jlk}\!\!}\right\}\mathbb{E}\left\{e^{jn\pi\!\cos\theta_{j'l'k'}\!\!}\right\}\right]\nonumber\\
\overset{(b)}=&N+2\sum_{n=1}^{N-1}(N-n)J_0^2(n\pi)\nonumber\\
\overset{(c)}\rightarrow &N+2\sum_{n=1}^{N-1}\frac{N-n}{n\pi^2}\nonumber\\
=&N-\frac{2}{\pi^2}(N-1)+\frac{2N}{\pi^2}\sum_{n=1}^{N-1}\frac{1}{n}\nonumber\\
\overset{(d)}\rightarrow&N-\frac{2}{\pi^2}(N-1)+\frac{2N}{\pi^2}(\ln N +a),
\end{align}
where $(a)-(d)$ use the similar manipulations as the corresponding derivations in \eqref{eta1}.
\end{proof}

\begin{lemma}
\label{lemma_eta3}
For the three independent channel vectors $\textbf{h}_{B,jlk}$, $\textbf{h}_{B,j'l'k'}$ and $\textbf{h}_{B,j''l''k''}$, where $(j,l,k)\neq (j',l',k')$, $(j,l,k)\neq (j'',l'',k'')$, and $(j',l',k')\neq (j'',l'',k'')$, it follows for large $N$ that
\begin{equation}
\begin{aligned}
&\mathbb{E}\left\{\textbf{h}_{B,jlk}^H \textbf{h}_{B,j'l'k'}\textbf{h}_{B,j'l'k'}^H  \textbf{h}_{B,j''l''k''}\right \}\rightarrow \eta_3,
\end{aligned}
\end{equation}
where $\eta_3\triangleq \eta_1+2\sum\limits_{m=1}^{N-1}\sum\limits_{n=0}^{N-m-1} J_0(m\pi)J_0(n\pi)J_0((n+m)\pi)$ and $\eta_1$ is defined in Lemma~\ref{lemma_eta1}.
\end{lemma}

\begin{proof}
Using \eqref{H_B},~$\mathbb{E}\left\{\textbf{h}_{B,jlk}^H \textbf{h}_{B,j'l'k'}\textbf{h}_{B,j'l'k'}^H  \textbf{h}_{B,j''l''k''}\right \}$ can be evaluated as
\begin{align}
&\mathbb{E}\left\{\textbf{h}_{B,jlk}^H \textbf{h}_{B,j'l'k'}\textbf{h}_{B,j'l'k'}^H  \textbf{h}_{B,j''l''k''}\right \}\nonumber\\
=&\mathbb{E}\left\{\!\sum_{n=0}^{N-1}\!e^{jn\pi\!(\!\cos\theta_{jlk}\!-\!\cos\theta_{j'l'k'}\!)\!}  \sum_{n=0}^{N-1}\!e^{jn\pi\!(\!\cos\theta_{j'l'k'}\!-\!\cos\theta_{j''l''k''}\!)\!}\!\right\}\nonumber\\
=&\mathbb{E}_{\theta_{j'l'k'}}\left\{  \mathbb{E}_{\theta_{jlk},\theta_{j''l''k''}|\theta_{j'l'k'}}\left\{    \!\sum_{n=0}^{N-1}e^{jn\pi\!(\!\cos\theta_{jlk}\!-\!\cos\theta_{j'l'k'}\!)\!} \!\right.\right.\nonumber\\
&\left.\left.~~~~~~~~~~~~~~~~~~~~~~~~~\times\! \sum_{n=0}^{N-1}e^{jn\pi\!(\!\cos\theta_{j'l'k'}\!-\!\cos\theta_{j''l''k''}\!)\!}    \right\} \right\}\nonumber \\
\overset{(a)}=&\mathbb{E}_{\theta_{j'l'k'}}\left\{  \sum_{n=0}^{N-1}e^{-jn\pi\!\cos\theta_{j'l'k'}\!} \mathbb{E}_{\theta_{jlk}}\left\{  e^{jn\pi\!(\!\cos\theta_{jlk})}    \right\}\right.\nonumber\\
&\left.~~~~~~~\times \sum_{n=0}^{N-1}e^{jn\pi\!\cos\theta_{j'l'k'}\!} \mathbb{E}_{\theta_{j''l''k''}}\left\{  e^{-jn\pi\!(\!\cos\theta_{j''l''k''})}    \right\}\right\}\nonumber\\
\overset{(b)}=&\mathbb{E}_{\theta_{j'l'k'}}\! \! \left\{ \! \sum_{n=0}^{N-1}\!\!J_0(n\pi)e^{\!-\!jn\pi\!\cos\theta_{j'l'k'}\!} \!   \sum_{n=0}^{N-1}J_0(n\pi)e^{jn\pi\!\cos\theta_{j'l'k'}\!} \!\!\right\}\nonumber
\end{align}
\begin{align}
=&\mathbb{E}_{\theta_{j'l'k'}}\! \! \left\{ \sum_{n=0}^{N-1}\! J_0^2(n\pi)\!+\! \sum_{m=1}^{N-1}\!\left(\!e^{\!-\!jm\pi\!\cos\theta_{j'l'k'}\!}+e^{jm\pi\!\cos\theta_{j'l'k'}\!}\!\right)\!\right.\nonumber \\
 &\left.~~~~~~~~~~~~~~~~~~~~~~~~~~\times \sum_{n=0}^{N-m-1} \!  J_0(n\pi)J_0((n+m)\pi)       \right\}\nonumber\\
\overset{(c)}=& \sum_{n=0}^{N-1}\! J_0^2(n\pi)\!+\!2 \sum_{m=1}^{N-1}  \sum_{n=0}^{N-m-1} \!J_0(m\pi)  J_0(n\pi)J_0((n+m)\pi) \nonumber\\
\overset{(d)} \rightarrow & \eta_1 +  2 \sum_{m=1}^{N-1}  \sum_{n=0}^{N-m-1} \!J_0(m\pi)  J_0(n\pi)J_0((n+m)\pi)  ,
\end{align}
where $(a)$ utilizes the fact that $\theta_{jlk}$ is independent of $\theta_{j''l''k''}$,
$(b)$ and $(c)$ come from $\mathbb{E}\left\{\!e^{-jn\pi\cos\theta_{j''l''k''}}\!\right\}$
$=\mathbb{E}\left\{\!e^{jn\pi\cos\theta_{jlk}}\!\right\}=\mathbb{E}\left\{\!e^{jn\pi\cos\theta_{j'l'k'}}\!\right\}=\mathbb{E}\left\{\!e^{-jn\pi\cos\theta_{j'l'k'}}\!\right\}$$=J_0(n\pi)$ demonstrated in \eqref{J},
and $(d)$ uses \eqref{eta1}.
\end{proof}

\section{Lemma 5}
\begin{lemma}
\label{lemma_ortho}
The mmWave MIMO channel matrix $\bar{\textbf{H}}_{jl}$ in \eqref{Heq} is asymptotically orthogonal with large $N$. Letting $N\rightarrow \infty$, we have
\begin{equation}
\begin{aligned}
\label{Heq_or}
\mathbb{E}\left\{\frac{1}{N}\bar{\textbf{H}}_{jl}^H \bar{\textbf{H}}_{jl}\right\}\rightarrow \textbf{B}_{jl} \textbf{C}_{jl}^H\textbf{C}_{jl}.
\end{aligned}
\end{equation}
\end{lemma}

\begin{proof}
Assume that $\beta_{jlk}=c_{jlk}=1$ for $j,l=1,2,...,L$ and $k=1,2,...,K$ for brevity, which does not affect the orthogonality of channel matrix $\bar{\textbf{H}}_{jl}$.
Substituting \eqref{H_B} to \eqref{Heq} and applying \emph{Lemma~\ref{lemma_eta1}}, we have
\begin{equation}
\begin{aligned}
\label{Heq_or_2}
\mathbb{E}\left\{\frac{1}{N}\bar{\textbf{H}}_{jl}^H \bar{\textbf{H}}_{jl}\right\}\rightarrow
\begin{bmatrix}
1  &  \frac{\eta_1}{N}  & \cdots\ & \frac{\eta_1}{N}\\
\frac{\eta_1}{N}  &  1  & \cdots\ & \frac{\eta_1}{N}\\
 \vdots   & \vdots & \ddots  & \vdots  \\
 \frac{\eta_1}{N} & \frac{\eta_1}{N}  & \cdots\ & 1\\
\end{bmatrix}.
\end{aligned}
\end{equation}
As $N$ goes to infinity, the non-diagonal elements of the above matrix asymptotically converge as
\begin{equation}
\begin{aligned}
\label{Heq_or_3}
\frac{\eta_1}{N}= \frac{1}{N}+\frac{\ln N+a}{N\pi^2} \rightarrow 0.
\end{aligned}
\end{equation}
Thus, we have
\begin{equation}
\label{Heq_or_4}
\mathbb{E}\left\{\frac{1}{N}\bar{\textbf{H}}_{jl}^H \bar{\textbf{H}}_{jl}\right\}\rightarrow  \textbf{I}_{K}.
\end{equation}
For general but finite values of $\beta_{jlk}$ and $c_{jlk}$, the orthogonality can be similarly established, which proves \emph{Lemma 5}.
\end{proof}

\section{Proof of Theorem 1}
\begin{proof}
In this appendix, we give the proof of \emph{Theorem~\ref{theorem_rate}} by applying \emph{Lemma~\ref{lemma_c}-\ref{lemma_eta3}} in Appendices B-C.
Start with deriving the lower bound for $\gamma$ in \eqref{gamma}.
Substituting \eqref{beta} into \eqref{Sd1} simplifies $S$ and yields
\begin{equation}
\begin{aligned}
\label{Sd2}
\frac{S}{|c_{jjk}|^4}=(1-\rho_{AD})^2P_tN^2.
\end{aligned}
\end{equation}
Before considering $\mathbb{E}\{|I_n|^2\}$, $\mathbb{E}\{|I_q|^2\}$, and $\mathbb{E}\{|S_r|^2\}$ in \eqref{I}, we introduce two definitions used for notational brevity.
Firstly, $\lambda$ is defined as
\begin{align}
&\frac{1}{|c_{jjk}|^2} \sum_{l=1}^L  \sum_{i=1}^K \beta_{jli} |c_{jli}|^2 \nonumber\\
\overset{(a)}=&1+\sum_{i\neq k} \frac{|c_{jji}|^2}{|c_{jjk}|^2} + \beta \sum_{l\neq j}  \sum_{i=1}^K \frac{|c_{jli}|^2}{|c_{jjk}|^2} \nonumber\\
\overset{(b)}\leq & c^{-2} [c^2+ (K-1)M+ \beta (L-1)KM] \nonumber\\
\triangleq & c^{-2} \lambda,
\label{lambda}
\end{align}
where $(a)$ utilizes \eqref{beta} and $(b)$ applies \emph{Lemma~\ref{lemma_c}}.
Secondly, for $\mu_j$ defined in \eqref{mu}, substitute \eqref{C_npq} and use \emph{Lemma~\ref{lemma_c}} and it yields
\begin{align}
&\frac{\mu_j }{|c_{jjk}|^2}\nonumber\\
=&\frac{\sigma_n^2}{(1\!-\!\rho_{AD}) P_p |c_{jjk}|^2}\!+\!\frac{\rho_{AD}}{(1-\rho_{AD})\tau |c_{jjk}|^2}  \sum_{l=1}^L  \sum_{i=1}^K \beta_{jli} |c_{jli}|^2\nonumber\\
\leq & c^{-2} \left[\frac{\sigma_n^2 }{(1-\rho_{AD}) P_p}+\frac{\rho_{AD} \lambda}{(1-\rho_{AD})\tau } \right]\nonumber\\
\triangleq & c^{-2}\mu,
\label{mu2}
\end{align}
where the inequality uses \eqref{lambda} and \emph{Lemma~\ref{lemma_c}}.

Now from \eqref{In}, an upper bound to $\frac{1}{|c_{jjk}|^4} \mathbb{E}\{|I_n|^2\}$ is obtained as
\begin{align}
&\frac{1}{|c_{jjk}|^4}\mathbb{E}\{|I_n|^2\}\nonumber\\
=&\frac{1}{|c_{jjk}|^4} (1-\rho_{AD}\!)^2\sigma_n^2 \left(\! N\mu_j \!+\!N \sum_{l=1}^L \beta_{jlk} |c_{jlk}|^2
\right.
\nonumber\\
&~~~~~~~~~~~~~~\left.+\sum_{l=1}^L \sum_{t\neq l} \beta_{jlk}^{\frac{1}{2}}\beta_{jtk}^{\frac{1}{2}}c_{jlk}^*c_{jtk} \textbf{h}_{B,jlk}^H\textbf{h}_{B,jtk}
\!\right)\nonumber\\
\overset{(a)}\leq &\frac{1}{|c_{jjk}|^2}(1-\rho_{AD})^2\sigma_n^2
\left(\! N\mu c^{-2} \!+\!\frac{N}{|c_{jjk}|^2} \sum_{l=1}^L \beta_{jlk} |c_{jlk}|^2
\right.
\nonumber\\
&~~~~~~~~~~\left.+\frac{1}{|c_{jjk}|^2} \sum_{l=1}^L \sum_{t\neq l} \beta_{jlk}^{\frac{1}{2}}\beta_{jtk}^{\frac{1}{2}}c_{jlk}^*c_{jtk} \textbf{h}_{B,jlk}^H\textbf{h}_{B,jtk}
\!\right)\nonumber\\
\overset{(b)}\leq & c^{-4}(1-\rho_{AD})^2\sigma_n^2\nonumber\\
&\times \!
\left[\!N\mu\!+\!Nc^2\!+\!N(\!L\!-\!1\!)\beta M \!+\!\sum_{l\neq j}\! \sum_{\substack{t\neq j\\t\neq l}}\! \beta M \textbf{h}_{B,jlk}^H\textbf{h}_{B,jtk} \right.\nonumber\\
&\left.~~~+ \beta^{\frac{1}{2}} c M ^{\frac{1}{2}} \left( \sum_{t\neq j} \textbf{h}_{B,jjk}^H\textbf{h}_{B,jtk} + \sum_{l\neq j} \textbf{h}_{B,jlk}^H\textbf{h}_{B,jjk}  \right) \right]\nonumber\\
\triangleq & I_{UB,n}
,
\label{I_UBn}
\end{align}
where $(a)$ uses \eqref{mu2} and $(b)$ utilizes \emph{Lemma~\ref{lemma_c}} and \eqref{beta}.

Similarly, by substituting \eqref{Iq}, an upper bound to $\frac{1}{|c_{jjk}|^4} \mathbb{E}\{|I_q|^2\}$ is derived as
\begin{align}
&\frac{1}{|c_{jjk}|^4}\mathbb{E}\{|I_q|^2\}\nonumber\\
=&\frac{1}{|c_{jjk}|^4} \sigma_{q,j}^2 \left(\! N\mu_j \!+\!N \sum_{l=1}^L \beta_{jlk} |c_{jlk}|^2
\right.
\nonumber\\
&~~~~~~~~~~~~~~~~\left.+\sum_{l=1}^L \sum_{t\neq l} \beta_{jlk}^{\frac{1}{2}}\beta_{jtk}^{\frac{1}{2}}c_{jlk}^*c_{jtk} \textbf{h}_{B,jlk}^H\textbf{h}_{B,jtk}
\!\right)\nonumber\\
\overset{(a)}\leq & \frac{c^{-2}}{|c_{jjk}|^2}  \rho_{AD}(1-\!\rho_{AD})(\sigma_n^2\!+\!\lambda P_t) \left(\! N\mu_j \!+\!N \sum_{l=1}^L \beta_{jlk} |c_{jlk}|^2
\right.
\nonumber\\
&~~~~~~~~~~~~~~~~\left.+\sum_{l=1}^L \sum_{t\neq l} \beta_{jlk}^{\frac{1}{2}}\beta_{jtk}^{\frac{1}{2}}c_{jlk}^*c_{jtk} \textbf{h}_{B,jlk}^H\textbf{h}_{B,jtk}
\!\right)\nonumber\\
\overset{(b)}\leq &c^{-4}\rho_{AD}(1\!-\!\rho_{AD})(\sigma_n^2\!+\!\lambda P_t)\nonumber\\
&\times \!
\left[\!N\mu\!+\!Nc^2\!+\!N(\!L\!-\!1\!)\beta M \!+\!\sum_{l\neq j}\! \sum_{\substack{t\neq j\\t\neq l}}\! \beta M \textbf{h}_{B,jlk}^H\textbf{h}_{B,jtk} \right.\nonumber\\
&\left.~~~+ \beta^{\frac{1}{2}} c M ^{\frac{1}{2}} \left( \sum_{t\neq j} \textbf{h}_{B,jjk}^H\textbf{h}_{B,jtk} + \sum_{l\neq j} \textbf{h}_{B,jlk}^H\textbf{h}_{B,jjk}  \right) \right]\nonumber\\
\triangleq & I_{UB,q}
,
\label{I_UBq}
\end{align}
where $(a)$ utilizes the inequality that
\begin{align}
\frac{\sigma_{q,j}^2}{|c_{jjk}|^2}
\!=&\rho_{AD}(\!1\!-\!\rho_{AD}\!) \left(\!\frac{\sigma_n^2}{|c_{jjk}|^2}\!+\!\frac{P_t }{|c_{jjk}|^2} \sum_{l=1}^L\sum_{k=1}^K\beta_{jlk} |c_{jlk}|^2 \!\!\right)\nonumber\\
\leq& c^{-2}\rho_{AD}(1-\rho_{AD})(\sigma_n^2+\lambda P_t),
\label{sigma_q_2}
\end{align}
which uses \eqref{sigma_q} and the inequality utilizes \emph{Lemma~\ref{lemma_c}} and \eqref{lambda},
and $(b)$ utilizes the similar manipulations as in \eqref{I_UBn}.

For the term of $\mathbb{E}\{|S_r|^2\}$ in \eqref{I}, we use \eqref{H_B} to rewrite the expression in \eqref{Sr} as
\begin{align}
&\mathbb{E}\{|S_r|^2\}\nonumber\\
=&(1\!-\!\rho_{AD})^2P_t \left(\!\!\mu_j N \sum_{l=1}^L \sum_{i=1}^K\beta_{jli}|c_{jli}|^2
\!+\!N^2 \sum_{t=1}^L  \beta_{jtk}^2|c_{jtk}|^4
\right.\nonumber\\
&+\sum_{t=1}^L \sum_{(l,i)\neq (t,k)} \beta_{jtk}\beta_{jli}|c_{jtk}|^2 |c_{jli}|^2  |\textbf{h}_{B,jtk}^H  \textbf{h}_{B,jli}|^2
\nonumber\\
&+N \sum_{t=1}^L  \sum_{r\neq t}
\beta_{jtk}^{\frac{3}{2}}  \beta_{jrk}^{\frac{1}{2}} c_{jtk}^* |c_{jtk}|^2 c_{jrk}
\textbf{h}_{B,jtk}^H  \textbf{h}_{B,jrk}
\nonumber\\
&+N\sum_{t=1}^L \sum_{r\neq t}
\beta_{jtk}^{\frac{1}{2}}  \beta_{jrk}^{\frac{3}{2}} c_{jtk}^* |c_{jrk}|^2 c_{jrk}
\textbf{h}_{B,jtk}^H \textbf{h}_{B,jrk}
\nonumber\\
&+\sum_{t=1}^L  \sum_{r\neq t} \sum_{\substack{(l,i)\neq (t,k)\\ (l,i)\neq(r,k)}}
\beta_{jtk}^{\frac{1}{2}} \beta_{jli} \beta_{jrk}^{\frac{1}{2}} c_{jtk}^* |c_{jli}|^2 c_{jrk}
\nonumber\\
&~~~~~~~~~~~~~~~~~~~~~~~~~~\times
\textbf{h}_{B,jtk}^H \textbf{h}_{B,jli}\textbf{h}_{B,jli}^H  \textbf{h}_{B,jrk}
\Bigg)
,
\label{Sr2}
\end{align}
Then, by applying \emph{Lemma~\ref{lemma_c}} and substituting \eqref{beta}, \eqref{lambda}, and \eqref{mu2} into \eqref{Sr2}, an upper bound to $\frac{1}{|c_{jjk}|^4} \mathbb{E}\{|S_r|^2\}$ is obtained as follows
\begin{align}
&\frac{1}{|c_{jjk}|^4} \mathbb{E}\{|S_r|^2\}\nonumber\\
\leq &
(1-\rho_{AD})^2P_t c^{-4}\Bigg[\mu N \lambda+N^2c^4+N^2(L-1)\beta^2M^2
\nonumber\\
&+c^2 M \sum_{i\neq k}   |\textbf{h}_{B,jjk}^H  \textbf{h}_{B,jji}|^2
\!+\!\beta c^2 M \sum_{l\neq j} \sum_{i=1}^K   |\textbf{h}_{B,jjk}^H  \textbf{h}_{B,jli}|^2\nonumber\\
&+\beta M^2 \sum_{t\neq j} \!\sum_{i=1}^K  \! |\textbf{h}_{B,jtk}^H  \textbf{h}_{B,jji}|^2\nonumber\\
&+\beta^2 M^2  \sum_{t\neq j} \sum_{\substack{(l,i)\neq (t,k) \\ l\neq j}}  |\textbf{h}_{B,jtk}^H  \textbf{h}_{B,jli}|^2\nonumber\\
&+N \!\left(\!\beta^{\frac{1}{2}} c^3 M^{\frac{1}{2}}\!\!+\!\!\beta^{\frac{3}{2}} c M^{\frac{3}{2}}\!\right) \sum_{r\neq j}  \textbf{h}_{B,jjk}^H \textbf{h}_{B,jrk}\nonumber \\
&+N \left(\beta^{\frac{3}{2}} c M^{\frac{3}{2}} + \beta^{\frac{1}{2}} c^3 M^{\frac{1}{2}}\right)  \sum_{t\neq j} \textbf{h}_{B,jtk}^H  \textbf{h}_{B,jjk}\nonumber\\
&+2N \beta^2M^2 \sum_{t\neq j}  \sum_{\substack{r\neq j\\r\neq t}} \textbf{h}_{B,jtk}^H  \textbf{h}_{B,jrk}
+A
\Bigg]\nonumber\\
\triangleq & S_{UB,r}
,
\label{S_UBr}
\end{align}
where $A$ comes from the last term in \eqref{Sr2} which can be expressed as
\begin{align}
A=&\beta^{\frac{1}{2}} c M^{\frac{3}{2}} \sum_{r\neq j} \sum_{i\neq k}
 \textbf{h}_{B,jjk}^H \textbf{h}_{B,jji}\textbf{h}_{B,jji}^H  \textbf{h}_{B,jrk}\nonumber\\
&+\beta^{\frac{1}{2}} c M^{\frac{3}{2}} \sum_{t\neq j} \sum_{i\neq k}
 \textbf{h}_{B,tjk}^H \textbf{h}_{B,jji}\textbf{h}_{B,jji}^H  \textbf{h}_{B,jjk}\nonumber\\
&+ \beta^{\frac{3}{2}} c M^{\frac{3}{2}} \sum_{r\neq j} \sum_{\substack{(l,i)\neq (r,k)\\l\neq j}}
 \textbf{h}_{B,jjk}^H \textbf{h}_{B,jli}\textbf{h}_{B,jli}^H  \textbf{h}_{B,jrk} \nonumber\\
&+ \beta^{\frac{3}{2}} c M^{\frac{3}{2}} \sum_{t\neq j} \sum_{\substack{(l,i)\neq (r,k)\\l\neq j}}
 \textbf{h}_{B,tjk}^H \textbf{h}_{B,jli}\textbf{h}_{B,jli}^H  \textbf{h}_{B,jjk}\nonumber \\
&+\beta M^2 \sum_{t\neq j}  \sum_{\substack{r\neq j\\r\neq t}} \sum_{i=1}^K
\textbf{h}_{B,jtk}^H \textbf{h}_{B,jji}\textbf{h}_{B,jji}^H  \textbf{h}_{B,jrk} \nonumber\\
&+\beta^2 M^2 \sum_{t\neq j}  \sum_{\substack{r\neq j\\r\neq t}} \sum_{\substack{(l,i)\neq(t,k)\\(l,i)\neq(r,k)\\l\neq j}}\!
\textbf{h}_{B,jtk}^H \textbf{h}_{B,jli}\textbf{h}_{B,jli}^H  \textbf{h}_{B,jrk}
.
\label{A}
\end{align}
Thus, based on \eqref{I}, \eqref{Sd2}, \eqref{I_UBn}, \eqref{I_UBq}, and \eqref{S_UBr}, a lower bound to $\gamma$ in \eqref{gamma} is obtained as
\begin{align}
\gamma&=\frac{S/|c_{jjk}|^4}{I/|c_{jjk}|^4}\nonumber\\
&\geq\frac{(1-\rho_{AD})^2P_tN^2}{I_{UB,n}+I_{UB,q}+S_{UB,r}-S}\nonumber\\
&\triangleq\gamma_{LB}.
\label{gamma2}
\end{align}
Then, using the Jensen's inequality, a lower bound to the achievable rate in \eqref{rate1} is established
\begin{align}
R&=\mathbb{E}\left\{\log(1+\gamma)\right\}\nonumber\\
&\geq\mathbb{E}\left\{\log(1+\gamma_{LB})\right\}\nonumber\\
&\geq\log\left(1+\frac{(1-\rho_{AD})^2P_tN^2}{\mathbb{E}\{I_{UB,n}+I_{UB,q}+S_{UB,r}-S\}}\right).
\label{bound1}
\end{align}

Thus far, by further plugging the expectation values of $I_{UB,n}$ in \eqref{I_UBn}, $I_{UB,q}$ in \eqref{I_UBq}, and $S_{UB,r}$ in \eqref{S_UBr} from \emph{Lemma~\ref{lemma_eta1}-\ref{lemma_eta3}} in Appendix~C, we obtain the desired bound in \eqref{R_l}.
\end{proof}

\end{appendices}

\end{document}